\documentclass{article}

\usepackage{arxiv}

\usepackage[utf8]{inputenc} 
\usepackage[T1]{fontenc}    
\usepackage{hyperref}       
\usepackage{url}            
\usepackage{booktabs}       
\usepackage{amsfonts}       
\usepackage{nicefrac}       
\usepackage{microtype}      
\usepackage{lipsum}		
\usepackage{amsthm}
\usepackage{bussproofs}
\usepackage{rotating}
\usepackage{amsmath}
\usepackage{graphicx}
\usepackage{bussproofs}
\usepackage{amsmath}
\usepackage{amssymb}
\usepackage{amsthm,amsmath,amssymb,amsfonts}
\usepackage{float}
\usepackage{subfig}
\usepackage{tikz-cd}
\usepackage{tikz}
\usepackage{keycommand}
\usetikzlibrary{cd}
\newtheorem{theorem}{Theorem}
\newtheorem{corollary}{Corrollary}
\newtheorem{proposition}{Proposition}
\newtheorem{lemma}{Lemma}
\newtheorem{definition}{Definition}

\newtheorem{example}{Example}

\newenvironment{bprooftree}
{\leavevmode\hbox\bgroup}
{\DisplayProof\egroup}

\title{An alternative approach to the calculation of fundamental groups based on labeled natural deduction}

\author{
Tiago Mendon\c{c}a Lucena de Veras\\ 
 Departamento de Matem\'atica\\
 Universidade Federal Rural de Pernambuco\\
 Recife, Brasil \\
  \texttt{tiago.veras@ufrpe.br} \\
   \And
 Arthur F. Ramos \\
  Centro de Inform\'atica\\
 Universidade Federal de Pernambuco\\
  Recife, Brasil \\
  \texttt{afr@cin.ufpe.br} \\
 \And
  Ruy J. G. B. de Queiroz \\
  Centro de Inform\'atica\\
 Universidade Federal de Pernambuco\\
  Recife, Brasil \\
  \texttt{ruy@cin.ufpe.br} \\
\And
  Anjolina G. de Oliveira \\
  Centro de Inform\'atica\\
 Universidade Federal de Pernambuco\\
  Recife, Brasil \\
  \texttt{ago@cin.ufpe.br} \\
}

\begin{document}
\maketitle

	\begin{abstract}
       
      In this work, we use a labelled deduction system based on the concept of computational paths (sequence of rewrites) as equalities between two terms of the same type. We also define a term rewriting system that is used to make computations between these computational paths, establishing equalities between equalities. We use a labelled deduction system based on the concept of computational paths (sequence of rewrites) to obtain some results of algebraic topology and with support of the Seifet-Van Kampen Theorem we will  calculate, in a way less complex than the one made in mathematics \cite{Munkres} and  the technique of homotopy type theory \cite{hott}, the fundamental group of Klein Blottle $\mathbb{K}^2$, of the Torus $\mathbb{T}^2$ and Two holed Torus $\mathbb{M}_2=\mathbb{T}^2\# \mathbb{T}^2$ (the connected sum two torus).
      
 \end{abstract}


	\section{Introduction}\label{intro}

The identity type is arguably one of the most interesting entities of  Martin-L\"{o}f type theory. From any type $A$, it is possible to construct the identity type $Id_A (x,y)$. This type establishes the relation of identity between two terms of $A$, i.e., if there is  $x =_p y: A$, then $p$ is a witness or proof that $x$ is indeed equal to $y$.  The proposal of the Univalence Axiom made the identity type one of the most studied aspects of type theory. It proposes that $x=y$ is equivalent to saying that $x\simeq y$, that is, the identity is an equivalence of equivalences. Another important aspect is the fact that it is possible to interpret the as paths between two points of the same space. This interpretation gives rise to the interesting interpretation of equality as a collection of homotopical paths. This connection of type theory and homotopy theory makes type theory a suitable foundation for both computation and mathematics. Nevertheless, this interpretation is only a semantical one \cite{email} and it was not proposed with a syntatical counterpart for the concept of path in type theory. For that reason, the addition of paths to the syntax of homotopy type theory has been recently proposed by De Queiroz, Ramos and De Oliveira \cite{Ruy1,Art3}, in these works, the authors use an entity known as `computational path', proposed by De Queiroz and Gabbay in 1994 \cite{Ruy4}, and show that it can be used to formalize the identity type.

On the other hand, one of the main interesting points of the interpretation of logical connectives via deductive systems which use a labelling system is the clear separation between a functional calculus on the labels (the names that record the steps of the proof) and a logical calculus on the formulas \cite{Lof1,Ruy4}. Moreover, this interpretation has important applications. The works of \cite{Ruy1,Ruy4,Ruy5,RuyAnjolinaLivro} claim that the harmony that comes with this separation makes labelled natural deduction a suitable framework to study and develop a theory of equality for natural deduction. Take, for example, the following cases taken from the $\lambda$-calculus \cite{Ruy5}:

\begin{center}
$(\lambda x.(\lambda y.yx)(\lambda w.zw))v \rhd_{\eta} (\lambda x.(\lambda y.yx)z)v \rhd_{\beta} (\lambda y.yv)z \rhd_{\beta} zv$

$(\lambda x.(\lambda y.yx)(\lambda w.zw))v \rhd_{\beta} (\lambda x(\lambda w.zw)x)v \rhd_{\eta} (\lambda x.zx)v \rhd_{\beta} zv$
\end{center}

In the theory of the $\beta\eta$-equality of $\lambda$-calculus, we can indeed say that $(\lambda x.(\lambda y.yx)(\lambda w.zw))v$ is equal to $zv$. Moreover, as we can see above, we have at least two ways of obtaining these equalities. We can go further, and call $s$ the first sequence of \textit{rewrites} that establish that $(\lambda x.(\lambda y.yx)(\lambda w.zw))v$ is indeed equal to $zv$. The second one, for example, we can call $r$. Thus, we can say that this equality is established by $s$ and $r$. As we will see in this paper, we $s$ and $r$ are examples of an entity known as \textit{computational path}. 

Since we now have labels (computational paths) that establishes the equality between two terms, interesting questions might arise: is $s$ different of $r$ or are they normal forms of this equality proof? If $s$ is equal to $r$, how can we prove this? We can answer questions like this when we work in a labelled natural deduction framework. The idea is that we are not limited by the calculus on the formulas, but we can also define and work with rules that apply to the labels. That way, we can use these rules to formally establish the equality between these labels, i.e., establish equalities between equalities. In this work, we will use a system proposed by \cite{Anjo1} and known as $LND_{EQ}$-$TRS$.

In that context, the contribution of this paper will be to propose a surprising connection: it is possible to use a labelled natural deduction system based on the concept of computational paths (sequence of rewrites) together with $LND_{EQ}$-$TRS$ to obtain some results of algebraic topology and study of fundamental groups of some surfaces with the support of the Seifert-Van Kampen Theorem.

Indeed, in this paper we will develop a theory and show that it is powerful enough to calculate the fundamental group of a circle, torus and real projective plane. For this, e use a labelled deduction system based on the concept of computational paths (sequence of rewrites).  Taking into account that in mathematics \cite{Munkres} the calculation of this fundamental group is quite laborious, we believe our work accomplishes this calculation in a less complex form. Nevertheless, to obtain this result we need to first formally define the concept of computational paths and define $LND_{EQ}$-$TRS$.

  	\section{Computational Paths} \label{path}
	
	In this section, our objective is to give a brief introduction to the theory of computational paths. One should refer to \cite{Ruy1,Art3} for a detailed development of this theory. 
	
	A computational path is based on the idea that it is possible (and useful!) to formally define when two computational objects $a,b : A$ are equal. These two objects are equal if one can reach $b$ from $a$ by applying a sequence of axioms or rules of inference. Such a sequence of operations forms a path. Since it is essentially an operation between two computational objects, it is said that this path is a computational one. Also, an application of an axiom or an inference rule transforms (or rewrites) a term into another. For that reason, a computational path is also known as a sequence of rewrites. Nevertheless, before we define formally a computational path, we can take a look at the rather standard equality theory, the $\lambda\beta\eta-equality$ \cite{lambda}:
	\begin{definition}
		The \emph{$\lambda\beta\eta$-equality} is composed by the following axioms:
		\begin{enumerate}
			\item[$(\alpha)$] $\lambda x.M = \lambda y.M[y/x]$ \quad if $y \notin FV(M)$;
			\item[$(\beta)$] $(\lambda x.M)N = M[N/x]$;
			\item[$(\rho)$] $M = M$;
			\item[$(\eta)$] $(\lambda x.Mx) = M$ \quad $(x \notin FV(M))$.
		\end{enumerate}
		And the following rules of inference:
		\bigskip
		\noindent
		\begin{bprooftree}
			\AxiomC{$M = M'$ }
			\LeftLabel{$(\mu)$ \quad}
			\UnaryInfC{$NM = NM'$}
		\end{bprooftree}
		\begin{bprooftree}
			\AxiomC{$M = N$}
			\AxiomC{$N = P$}
			\LeftLabel{$(\tau)$}
			\BinaryInfC{$M = P$}
		\end{bprooftree}
		
		\bigskip
		\noindent
		\begin{bprooftree}
			\AxiomC{$M = M'$ }
			\LeftLabel{$(\nu)$ \quad}
			\UnaryInfC{$MN = M'N$}
		\end{bprooftree}
		\begin{bprooftree}
			\AxiomC{$M = N$}
			\LeftLabel{$(\sigma)$}
			\UnaryInfC{$N = M$}
		\end{bprooftree}
		
		\bigskip
		\noindent
		\begin{bprooftree}
			\AxiomC{$M = M'$ }
			\LeftLabel{$(\xi)$ \quad}
			\UnaryInfC{$\lambda x.M= \lambda x.M'$}
		\end{bprooftree}
		
		
		
	\end{definition}
	
	
	
	
	
	\begin{definition}($\beta$-equality \cite{lambda})
		$P$ is $\beta$-equal or $\beta$-convertible to $Q$  (notation $P=_\beta Q$)
		iff $Q$ is obtained from $P$ by a finite (perhaps empty)  series of $\beta$-contractions
		and reversed $\beta$-contractions  and changes of bound variables.  That is,
		$P=_\beta Q$ iff \textbf{there exist} $P_0, \ldots, P_n$ ($n\geq 0$)  such that
		$P_0\equiv P$,  $P_n\equiv Q$,
		$(\forall i\leq n-1) (P_i\triangleright_{1\beta}P_{i+1}  \mbox{ or }P_{i+1}\triangleright_{1\beta}P_i  \mbox{ or } P_i\equiv_\alpha P_{i+1}).$
	\end{definition}

	The same happens with $\lambda\beta\eta$-equality:
	\begin{definition}[$\lambda\beta\eta$-equality \cite{lambda}]
		The equality-relation determined by the theory $\lambda\beta\eta$ is called
		$=_{\beta\eta}$; that is, we define
		$$M=_{\beta\eta}N\quad\Leftrightarrow\quad\lambda\beta\eta\vdash M=N.$$
	\end{definition}
	
	\begin{example}
		Take the term $M\equiv(\lambda x.(\lambda y.yx)(\lambda w.zw))v$. Then, it is $\beta\eta$-equal to $N\equiv zv$ because of the sequence:\\
		$(\lambda x.(\lambda y.yx)(\lambda w.zw))v, \quad  (\lambda x.(\lambda y.yx)z)v, \quad   (\lambda y.yv)z , \quad zv$\\
		which starts from $M$ and ends with $N$, and each member of the sequence is obtained via 1-step $\beta$-contraction or $\eta$-contraction of a previous term in the sequence. To take this sequence into a {\em path\/}, one has to apply transitivity twice, as we will see below. Taking this sequence into a path leads us to the following:
	
		\noindent The first is equal to the second based on the grounds:\\
		$\eta((\lambda x.(\lambda y.yx)(\lambda w.zw))v,(\lambda x.(\lambda y.yx)z)v)$\\
		The second is equal to the third based on the grounds:\\
		$\beta((\lambda x.(\lambda y.yx)z)v,(\lambda y.yv)z)$\\
		Now, the first is equal to the third based on the grounds:\\
		$\tau(\eta((\lambda x.(\lambda y.yx)(\lambda w.zw))v,(\lambda x.(\lambda y.yx)z)v),\beta((\lambda x.(\lambda y.yx)z)v,(\lambda y.yv)z))$\\
		Now, the third is equal to the fourth one based on the grounds:\\
		$\beta((\lambda y.yv)z,zv)$\\
		Thus, the first one is equal to the fourth one based on the grounds:\\
		$\tau(\tau(\eta((\lambda x.(\lambda y.yx)(\lambda w.zw))v,(\lambda x.(\lambda y.yx)z)v),\beta((\lambda x.(\lambda y.yx)z)v,(\lambda y.yv)z)),\beta((\lambda y.yv)z,zv)))$.
	\end{example}
	
	
	The aforementioned theory establishes the equality between two $\lambda$-terms. Since we are working with computational objects as terms of a type, we need to translate the $\lambda\beta\eta$-equality to a suitable equality theory based on Martin L\"of's type theory. For the $\Pi$-type, for example, we obtain:
	
	\begin{definition}
		The equality theory of Martin L\"of's type theory has the following basic proof rules for the $\Pi$-type \cite{Ruy1,Art3}:
		
		\bigskip
		
		\noindent
		\begin{bprooftree}
			\hskip -0.3pt
			\alwaysNoLine
			\AxiomC{$N : A$}
			\AxiomC{$[x : A]$}
			\UnaryInfC{$M : B$}
			\alwaysSingleLine
			\LeftLabel{$(\beta$) \quad}
			\BinaryInfC{$(\lambda x.M)N = M[N/x] : B[N/x]$}
		\end{bprooftree}
		\begin{bprooftree}
			\hskip 17pt
			\alwaysNoLine
			\AxiomC{$[x : A]$}
			\UnaryInfC{$M = M' : B$}
			\alwaysSingleLine
			\LeftLabel{$(\xi)$ \quad}
			\UnaryInfC{$\lambda x.M = \lambda x.M' : \Pi_{(x : A)}B$}
		\end{bprooftree}
		
		\bigskip
		
		\noindent
		\begin{bprooftree}
			\hskip -0.5pt
			\AxiomC{$M : A$}
			\LeftLabel{$(\rho)$ \quad}
			\UnaryInfC{$M = M : A$}
		\end{bprooftree}
		\begin{bprooftree}
			\hskip 100pt
			\AxiomC{$M = M' : A$}
			\AxiomC{$N : \Pi_{(x : A)}B$}
			\LeftLabel{$(\mu)$ \quad}
			\BinaryInfC{$NM = NM' : B[M/x]$}
		\end{bprooftree}
		
		\bigskip
		
		\noindent
		\begin{bprooftree}
			\hskip -0.5pt
			\AxiomC{$M = N : A$}
			\LeftLabel{$(\sigma) \quad$}
			\UnaryInfC{$N = M : A$}
		\end{bprooftree}
		\begin{bprooftree}
			\hskip 105pt
			\AxiomC{$N : A$}
			\AxiomC{$M = M' : \Pi_{(x : A)}B$}
			\LeftLabel{$(\nu)$ \quad}
			\BinaryInfC{$MN = M'N : B[N/x]$}
		\end{bprooftree}
		
		\bigskip
		
		\noindent
		\begin{bprooftree}
			\hskip -0.5pt
			\AxiomC{$M = N : A$}
			\AxiomC{$N = P : A$}
			\LeftLabel{$(\tau)$ \quad}
			\BinaryInfC{$M = P : A$}
		\end{bprooftree}
		
		\bigskip
		
		\noindent
		\begin{bprooftree}
			\hskip -0.5pt
			\AxiomC{$M: \Pi_{(x : A)}B$}
			\LeftLabel{$(\eta)$ \quad}
			\RightLabel {$(x \notin FV(M))$}
			\UnaryInfC{$(\lambda x.Mx) = M: \Pi_{(x : A)}B$}
		\end{bprooftree}
		
		\bigskip
		
	\end{definition}
	
	We are finally able to formally define computational paths:
	
	\begin{definition}
		Let $a$ and $b$ be elements of a type $A$. Then, a \emph{computational path} $s$ from $a$ to $b$ is a composition of rewrites (each rewrite is an application of an inference rule of the equality theory of type theory or is a change of bound variables). We denote that by $a =_{s} b$.
	\end{definition}
	
	As we have seen in \textbf{example 1}, compositions of rewrites are applications of the rule $\tau$. Since change of bound variables is possible, each term is considered up to $\alpha$-equivalence.
	
	\section{A Term Rewriting System for Paths}
	
	As we have just shown, a computational path establishes when two terms of the same type are equal. From the theory of computational paths, an interesting case arises. Suppose we have a path $s$ that establishes that $a =_{s} b : A$ and a path $t$ that establishes that $a =_{t} b : A$. Consider that $s$ and $t$ are formed by distinct compositions of rewrites. Is it possible to conclude that there are cases that $s$ and $t$ should be considered equivalent? The answer is \emph{yes}. Consider the following example:
	
	\begin{example}
		\noindent \normalfont Consider the path  $a =_{t} b : A$. By the symmetry property, we obtain $b =_{\sigma(t)} a : A$. What if we apply the property again on the path $\sigma(t)$? We would obtain a path  $a =_{\sigma(\sigma(t))} b : A$. Since we applied symmetry twice in succession, we obtained a path that is equivalent to the initial path $t$. For that reason, we conclude the act of applying symmetry twice in succession is a redundancy. We say that the path $\sigma(\sigma(t))$ can be reduced to the path $t$.
	\end{example}
	
	As one could see in the aforementioned example, different paths should be considered equal if one is just a redundant form of the other. The example that we have just seen is just a straightforward and simple case. Since the equality theory has a total of 7 axioms, the possibility of combinations that could generate redundancies are rather high. Fortunately, most possible redundancies were thoroughly mapped by \cite{Anjo1}. In that work, a system that establishes redundancies and creates rules that solve them was proposed. This system, known as $LND_{EQ}$-$TRS$, originally mapped a total of 39 rules. For each rule, there is a proof tree that constructs it.  We included all rules in \textbf{appendix B}. To illustrate those rules, take the case of \textbf{example 2}. We have the following \cite{Ruy1}:
	
	\begin{prooftree}
		\AxiomC{$x =_{t} y : A$}
		\UnaryInfC{$y =_{\sigma(t)} x : A$}
		\RightLabel{\quad $\rhd_{ss}$ \quad $x =_{t} y : A$}
		\UnaryInfC{$x =_{\sigma(\sigma(t))} y : A$}
	\end{prooftree}
	
	\bigskip
	
	It is important to notice that we assign a label to every rule. In the previous case, we assigned the label $ss$. 
	\begin{definition}[$rw$-rule \cite{Art3}]
		\normalfont An $rw$-rule is any of the rules defined in $LND_{EQ}$-$TRS$.
	\end{definition}
	
	\begin{definition}[$rw$-contraction \cite{Art3}]
		Let $s$ and $t$ be computational paths. We say that $s \rhd_{1rw} t$ (read as: $s$ $rw$-contracts to $t$) iff we can obtain $t$ from $s$ by an application of only one $rw$-rule. If $s$ can be reduced to $t$ by finite number of $rw$-contractions, then we say that $s \rhd_{rw} t$ (read as $s$ $rw$-reduces to $t$).
		
	\end{definition}
	
	\begin{definition}[$rw$-equality \cite{Art3}]
		\normalfont  Let $s$ and $t$ be computational paths. We say that $s =_{rw} t$ (read as: $s$ is $rw$-equal to $t$) iff $t$ can be obtained from $s$ by a finite (perhaps empty) series of $rw$-contractions and reversed $rw$-contractions. In other words, $s =_{rw} t$ iff there exists a sequence $R_{0},....,R_{n}$, with $n \geq 0$, such that
		
		\centering $(\forall i \leq n - 1) (R_{i}\rhd_{1rw} R_{i+1}$ or $R_{i+1} \rhd_{1rw} R_{i})$
		
		\centering  $R_{0} \equiv s$, \quad $R_{n} \equiv t$
	\end{definition}
	
	\begin{proposition}\label{proposition3.7} $rw$-equality  is transitive, symmetric and reflexive.
	\end{proposition}
	
	\begin{proof}
		It follows directly from the fact that $rw$-equality is the transitive, reflexive and symmetric closure of $rw$.
	\end{proof}
	
	The above proposition is rather important, since sometimes we want to work with paths up to $rw$-equality. For example, we can take a path $s$ and use it as a representative of an equivalence class, denoting this by $[s]_{rw}$.
	
	We'd like to mention that  $LND_{EQ}$-$TRS$ is terminating and confluent. The proof of this can be found in \cite{Anjo1,Ruy2,Ruy3,RuyAnjolinaLivro}.
	
	One should refer to \cite{Ruy1,RuyAnjolinaLivro} for a more complete and detailed explanation of the rules of $LND_{EQ}$-$TRS$.

	
	
	\section{Fundamental Group of surfaces and results of algebraic topology obtained by means of Computational Paths}
	
	The objective of this section is to obtain the fundamental group of the surfaces like Klein Bottle, Torus and Two-Holed Torus by means of computational paths together with Seifert-Van Kampen Theorem. This calculation will be carried out in a less complex way than that used in Homotopy Type Theory and the one used in Maths. 
	
	
	Therefore, in the next subsection we will ensure that some algebraic topology results are valid using computational paths, we will also need of the deformation retract definition and the statement of the Seifert-Van Kampen Theorem. In the following subsection we will calculate the fundamental group of the Klein bottle, of the torus and the connected sum of two tori.

\subsection{Proof of results of algebraic topology by computational paths}	

In this subsection we will prove some results of the algebraic topology by means of computational paths. Such results are indispensable for obtaining our main result which is to get the fundamental group of the Klein bottle. Furthermore, these proofs establish, even more, computational paths as a working tool.

The next definition is necessary from the proof of the theorem that happens. The topological result of this theorem is central to the conclusion of the fundamental group of the Klein bottle.

\begin{definition}
 Let $X$ be a space connected by paths and $x_{0}\underset{\mu}{=} x_{1}$, that is, $\mu$ is a path between $x_0$ and $x_1$ where $x_0, x_1 \in X$.
\end{definition}

Define the map  $$\varphi_{\mu}: \pi_1(X,x_0) \longrightarrow \pi_1(X,x_1)$$ given by: $$\varphi_{\mu}([\alpha_{x_0)}])=[\sigma(\mu)]*[\alpha_{x_0}]*[\mu]= \tau \big(\tau \big(\sigma(\mu),\alpha_{x_0}\big),\mu \big),$$

where $\alpha_{x_0}$ is a loop in $X$ with base point $x_0$. See that $\varphi_{\mu}$ is well defined because  the map $\tau$ is well defined. So  $\varphi_{\mu}([\alpha_{x_0}])$ is a loop in $X$ with base point $x_1$.

Therefore, the $\varphi_{\mu}$ map takes a $\alpha_{x_0} \in \pi_1(X,x_0)$ and takes in a $\varphi_{\mu}([\alpha_{x_0}]) \in \pi_1(X,x_1)$.
\begin{theorem}
The map $\varphi_{\mu}$ is an isomorphism, that is, $\pi_{1}(X,x_0) \simeq \pi_{1}(X,x_1)$.
\end{theorem}
\begin{proof}
Let $\alpha_{x_0}$ and $\beta_{x_0}$ be two loops with base point in $x_0\in X$, so:

\begin{eqnarray*}
    \varphi_{\mu}([\alpha_{x_0}])*\varphi_{\mu}([\beta_{x_0}])&=& \big([\sigma(\mu)]*[\alpha_{x_0}]*[\mu]\big)*\big([\sigma(\mu)]*[\beta_{x_0}]*[\mu]\big)\\
    &=&\tau \bigg( \tau \bigg(\tau \big(\sigma(\mu),\alpha_{x_0}\big),\mu \bigg), \tau \bigg(\tau \big(\sigma(\mu),\beta_{x_0}\big),\mu \bigg) \bigg) \\
    &\triangleright_{tt} &\tau \bigg( \tau \bigg(\tau \big(\sigma(\mu),\alpha_{x_0}\big),\mu \bigg), \tau \bigg(\sigma(\mu),\tau(\beta_{x_0},\mu ) \bigg) \bigg).
\end{eqnarray*}

We will use the rewrite rules   $$\tau(\tau(t,r),s) \triangleright_{tt} \tau(t,\tau (r,s)),$$

under the following conditions, put:
\begin{eqnarray*}
    t&=& (\tau \big(\sigma(\mu),\alpha_{x_0}\big) \\
    r&=&\mu\\
    s&=& \tau \big(\sigma(\mu),\beta_{x_0}\big),\mu \big).\\
\end{eqnarray*}
 So,
 
\begin{eqnarray*}
    &\triangleright_{tt} &\tau \bigg( \tau \bigg(\tau \big(\sigma(\mu),\alpha_{x_0}\big),\mu \bigg), \tau \bigg(\sigma(\mu),\tau(\beta_{x_0},\mu ) \bigg) \bigg)\\
    &\triangleright_{tt} &\tau\bigg( \tau \big(\sigma(\mu),\alpha_{x_0}\big), \tau \bigg(\mu,  \tau \bigg(\sigma(\mu),\tau(\beta_{x_0},\mu ) \bigg) \bigg)\\
    &\triangleright_{\sigma(tt)} &\tau\bigg( \tau \big(\sigma(\mu),\alpha_{x_0}\big), \tau \bigg( \tau \big(\mu,\sigma(\mu)\big),\tau(\beta_{x_0},\mu ) \bigg) \bigg)\\
    &\triangleright_{tst} &\tau\bigg( \tau \big(\sigma(\mu),\alpha_{x_0}\big), \tau \bigg( \rho,\tau(\beta_{x_0},\mu ) \bigg) \bigg)\\
    &\triangleright_{tlr} &\tau\bigg( \tau \big(\sigma(\mu),\alpha_{x_0}\big), \tau(\beta_{x_0},\mu ) \bigg)\\
    &\triangleright_{tt} &\tau\bigg( \sigma(\mu),\tau \bigg(\alpha_{x_0},\tau(\beta_{x_0},\mu )\bigg)\bigg)\\
    &\triangleright_{\sigma(tt)} &\tau\bigg( \tau\bigg(\sigma(\mu),\tau(\alpha_{x_0},\beta_{x_0})\bigg),\mu \bigg)\\
    &=&[\sigma(\mu)]*[\alpha_{x_0}*\beta_{x_0}]*[\mu]\\
    &=&\varphi_{\mu}([\alpha_{x_0}*\beta_{x_0}]).\\
\end{eqnarray*}

Therefore, $\varphi_{\mu}([\alpha_{x_0}])*\varphi_{\mu}([\alpha_{x_0}])=\varphi_{\mu}([\alpha_{x_0}*\beta_{x_0}])$ and then we have $\varphi_{\mu}$ is a homomorphism.

Now, we need to prove that there is an inverse map of $\varphi_{\mu}$. For this, let $\alpha_{x_1}$ be a loop with base point in $x_1\in X$ and set the map  $$\kappa_{\mu}: \pi_1(X,x_1) \longrightarrow \pi_1(X,x_0)$$ given by: $$\kappa_{\mu}([\alpha_{x_1)}])=[\mu]*[\alpha_{x_1}]*[\sigma(\mu])= \tau \big(\tau \big(\mu,\alpha_{x_1}\big),\sigma(\mu) \big).$$

Note that $\kappa([\alpha_{x_1}])$ is a loop with base point in $x_0$, so we can calculate:

\begin{eqnarray*}
    \varphi_{\mu}\big([\kappa_{\mu}([\alpha_{x_1}])]\big)&=& [\sigma(\mu)]*\big[[\mu]*[\alpha_{x_1}]*[\sigma(\mu)]\big]*[\mu]\\
    &=&\tau \bigg( \tau \bigg(\sigma(\mu),\tau \bigg(\tau (\mu,\alpha_{x_1}),\sigma(\mu)\bigg)\bigg), \mu \bigg) \\
    &\triangleright_{\sigma(tt)} &\tau \bigg( \tau \bigg(\tau \bigg(\sigma(\mu),\tau(\mu,\alpha_{x_1})\bigg), \sigma(\mu) \bigg),\mu \bigg)\\
     &\triangleright_{\sigma(tt)} &\tau \bigg( \tau \bigg( \tau\bigg(\tau \big(\sigma(\mu),\mu\big),\alpha_{x_1}\bigg), \sigma(\mu)\bigg),\mu \bigg)\\
    &\triangleright_{tsr} &\tau \bigg( \tau \bigg( \tau\big( \rho,\alpha_{x_1}\big), \sigma(\mu)\bigg),\mu \bigg)\\
    &\triangleright_{tlr} &\tau \bigg( \tau\big(\alpha_{x_1}, \sigma(\mu)\big),\mu \bigg)\\
    &\triangleright_{tt} &\tau \bigg( \alpha_{x_1},\tau \big( \sigma(\mu),\mu \big) \bigg)\\
    &\triangleright_{tsr} &\tau \big(\alpha_{x_1},\rho \big)\\ 
    &\triangleright_{tsr} &\alpha_{x_1}.
\end{eqnarray*}

We have then $ \varphi_{\mu}\big([\kappa_{\mu}([\alpha_{x_1}])]\big)=\alpha_{x_1}$, and in a similar way we can show that $ \kappa_{\mu}\big( [\varphi_{\mu}([\alpha_{x_0}])] \big) =\alpha_{x_0}$. This implies that $\kappa_{\mu}= \varphi_{\mu}^{-1}$, which proves our theorem.
\end{proof}

The theorem above says that if $X$ is a space connected by paths, for any two points $x_0, x_1 \in X$, we have that  $\pi_1(X,x_0) \simeq \pi_1(X,x_1)$.

\begin{definition}
 
 Let $X$ and $Y$ be spaces connected by paths, $\mu:(X,x_0) \rightarrow (Y,y_0)$ be a continuous map that carries the point $x_0 \in X$ to the point $y_0 \in Y$, that is, $x_{0}\underset{\mu}{=} y_{0}$. Define the map
 
 $$\mu_{*}:\pi_1(X,x_0) \longrightarrow \pi_{1}(Y,y_0)$$
 
 by 
 
 $$\mu_{*}([\alpha_{x_0}])= [\sigma(\mu)]*[\alpha_{x_0}]*[\mu]=\tau \big(\tau \big(\sigma(\mu),\alpha_{x_0}\big),\mu \big),$$
 
where $\alpha_{x_0}$ is a loop in $X$ with base point $x_0$,      $y_{0}\underset{\sigma(\mu)}{=} x_{0}$ and note that  $\mu_{*}([\alpha_{x_0}]) \in \pi_{1}(Y,y_0).$
\end{definition}

\begin{corollary}
    The map $\mu_*$ is a homomorphism. 
\end{corollary}

\begin{proof}

 The map $\mu_*$ is in fact a homomorphism, let $([\alpha_{x_0}],([\beta_{x_0}] \in \pi_1(X,x_0)$, we need to prove that $\mu_{*}([\alpha_{x_0}]*[\beta_{x_0}])=\mu_{*}([\alpha_{x_0}])*\mu_{*}([\beta_{x_0}])$.
 
\begin{eqnarray*}
\mu_{*}\Big([\alpha_{x_0}]*[\beta_{x_0}]\Big) &=& [\sigma(\mu)]*([\alpha_{x_0}]*[\beta_{x_0}])*[\mu]  \\
 &=&\tau \Bigg(\tau \bigg(\sigma(\mu),\tau\big(\alpha_{x_0},\beta_{x_0}\big)\bigg),\mu \Bigg) \\
 &\underset{\sigma(tt)}{=}&\tau \Bigg(\tau \bigg(\tau\big(\sigma(\mu),\alpha_{x_0}\big),\beta_{x_0}\bigg),\mu \Bigg)\\
 &\underset{\sigma(tlr)}{=}&\tau \Bigg(\tau \bigg(\tau\big(\sigma(\mu),\alpha_{x_0}\big),\tau\big(\rho_{x_0},\beta_{x_0}\big)\bigg),\mu \Bigg)\\
 &\underset{\sigma(tr)}{=}&\tau \Bigg(\tau \bigg(\tau\big(\sigma(\mu),\alpha_{x_0}\big),\tau\Big(\tau\big(\mu,\sigma(\mu)\big),\beta_{x_0}\Big)\bigg),\mu \Bigg)\\
 &\underset{tt}{=}&\tau \Bigg(\tau \bigg(\tau\big(\sigma(\mu),\alpha_{x_0}\big),\tau\Big(\mu,\tau\big(\sigma(\mu),\beta_{x_0}\big)\Big),\mu \Bigg)\\
 &\underset{\sigma(tt)}{=}&\tau\Bigg( \tau\bigg( \tau\Big(\tau \big(\sigma(\mu),\alpha_{x_0}\big),\mu\Big),\tau\big(\sigma(\mu),\beta_{x_0}\big)    \bigg) ,\mu \Bigg)\\
 &\underset{tt}{=}&\tau\Bigg( \tau\Big( \tau \big(\sigma(\mu),\alpha_{x_0}\big),\mu\Big), \tau\Big( \tau\big(\sigma(\mu),\beta_{x_0}\big),\mu\Big) \Bigg)\\
 &=& \mu_{*}\big([\alpha_{x_0}]\big)*\mu_{*}\big([\beta_{x_0}]\big). \\   
\end{eqnarray*}

In this way, we can conclude that $\mu_{*}$ is indeed a homomorphism and $\mu_{*}$ is called homomorphism induced by $\mu$.
\end{proof}

It follows as a consequence of the previous definition that if $i:(X,x_0) \rightarrow (X,x_0)$ is the identity map, then $i_{*}$ is the identity homomorphism. Just check that $i=\sigma(i)=\rho_{x_0}$ and

\begin{eqnarray*}
i_{*}\Big([\alpha_{x_0}]\Big) &=& [\sigma(i)]*[\alpha_{x_0}]*[i]  \\
 &=&\tau \Big(\tau \big(\sigma(i),\alpha_{x_0}\big),i \Big) \\
 &=&\tau \Big(\tau \big(\rho_{x_0},\alpha_{x_0}\big),\rho_{x_0} \Big) \\
 &\underset{tlr}{=}&\tau(\alpha_{x_0},\rho_{x_0})\\
 &\underset{trr}{=}&\alpha_{x_0}.\\
\end{eqnarray*}

\begin{theorem}
If $\mu:(X,x_0) \rightarrow (Y,y_0)$ and $\Delta:( Y,y_0) \rightarrow (Z,z_0)$ are continuous maps, then $(\mu*\Delta)_{*}=\mu_{*}*\Delta_{*}$. 
\end{theorem}

\begin{proof}
\begin{eqnarray*}
\Big(\mu_{*}*\Delta_{*}\Big)\big([\alpha_{x_0}]\big) &=& \Delta_{*}\Big(\Big[\mu_{*}\big([\alpha_{x_0}]\big)\Big]\Big)\\ &=&[\sigma(\Delta)]*\Big[\mu_{*}\big([\alpha_{x_0}]\big)\Big]*[\Delta]  \\ &=&[\sigma(\Delta)]*\Big[[\sigma(\mu)]*[\alpha_{x_0}]*[\mu])\Big]*[\Delta] \\
 &=&\tau \Bigg(\tau \bigg(\sigma(\Delta),\tau \Big(\tau(\sigma(\mu,)\alpha_{x_0}),\mu\Big)\bigg),\Delta \Bigg) \\
&\underset{\sigma(tt)}{=}&\tau \Bigg(\tau \bigg(\tau\Big(\sigma(\Delta), \tau(\sigma(\mu),\alpha_{x_0})\Big),\mu\bigg),\Delta \Bigg) \\
 &\underset{\sigma(tt)}{=}&\tau \Bigg(\tau \bigg(\tau\Big(\tau\big(\sigma(\Delta),\sigma(\mu)\big),\alpha_{x_0})\Big),\mu\bigg),\Delta \Bigg) \\
&\underset{tt}{=}&\tau\Bigg( \tau \bigg( \tau\big( \sigma(\Delta),\sigma(\mu)\big), \tau\big( \alpha_{x_0},\mu\big)\bigg),\Delta \Bigg) \\
&\underset{tt}{=}&\tau\Bigg(\tau\big( \sigma(\Delta),\sigma(\mu)\big), \tau\Big(\tau\big( \alpha_{x_0},\mu\big),\Delta\Big) \Bigg) \\
&\underset{tt}{=}&\tau\Bigg(\tau\big( \sigma(\Delta),\sigma(\mu)\big), \tau\Big( \alpha_{x_0},\tau\big(\mu,\Delta\big)\Big) \Bigg) \\
&\underset{\sigma(stss)}{=}&\tau\Bigg( \sigma\Big( \tau(\Delta,\mu)\Big), \tau\Big( \alpha_{x_0},\tau\big(\mu,\Delta\big)\Big) \Bigg) \\
&\underset{\sigma(tt)}{=}&\tau\Bigg( \tau\bigg(\sigma\Big( \tau(\Delta,\mu)\Big),\alpha_{x_0}\bigg),\tau\big(\mu,\Delta\big)\Bigg) \\
&=&\Big[ \sigma\Big(\tau(\Delta,\mu)\Big)\Big]*\Big[ \alpha_{x_0} \Big]*\Big[ \tau(\mu,\Delta) \Big]\\ 
&=&\Big[ \sigma\Big(\Delta*\mu\Big)\Big]*\Big[ \alpha_{x_0} \Big]*\Big[\mu*\Delta \Big]\\ 
&=&\Big( \mu*\Delta\Big)_{*}\big([\alpha_{x_0}]\big).
\end{eqnarray*}

\end{proof}

We have seen that if  $\mu:(X,x_0) \rightarrow (Y,y_0)$ is a continuous map, we obtain that $\mu_{*}$ is a homomorphism induced by $\mu$. But, in the particular case where $\mu$ is a homeomorphism, what can we say about the map $\mu_{*}$? This question will be answered with the next theorem that will be stated and demonstrated by computational paths. This result is also indispensable to fulfill our main objective.

\begin{theorem}
If $\mu:(X,x_0) \rightarrow (Y,y_0)$ is a homeomorphism of $X$ with $Y$, then $\mu_{*}:\pi_1(X,x_0) \rightarrow \pi_{1}(Y,y_0)$ is an isomorphism.
\end{theorem}

\begin{proof}

Let $\sigma(\mu)$ be the inverse of $\mu$, so we can write $\sigma(\mu):(Y,y_0) \rightarrow (X,x_0)$. We must prove that the map

$$\sigma(\mu)_{*}:\pi_{1}(Y,y_0) \rightarrow \pi_1(X,x_0)$$

given by:

$$\sigma(\mu)_{*}\big([\alpha_{y_0}]\big)=[\mu]*[\alpha_{y_0}]*[\sigma(\mu)]=\tau\Big(\tau(\mu,\alpha_{y_0}),\sigma(\mu)\Big).$$

Once $\mu_{*}([\alpha_{x_0}]) \in \pi_{1}(Y,y_0)$, we have

\begin{eqnarray*}
\sigma(\mu)_{*}\Big(\big[\mu_{*}([\alpha_{x_0}])\big]  \Big) &=&[\mu]*\big[\mu_{*}([\alpha_{x_0}])\big]*[\sigma(\mu)]\\
 &=&[\mu]*\big[[\sigma(\mu)]*[\alpha_{x_0}]*[\mu]\big]*[\sigma(\mu)] \\
  &=&\tau \Bigg( \tau\bigg(\mu,\tau\Big(\tau(\sigma(\mu),\alpha_{x_0}),\mu  \Big)    \bigg),\sigma(\mu)\Bigg)\\
   &\underset{\sigma(tt)}{=}&\tau \Bigg( \tau\bigg( \tau\Big(\mu,\tau\big(\sigma(\mu),\alpha_{x_0}\big)\Big),\mu    \bigg),\sigma(\mu)\Bigg)\\
    &\underset{tt}{=}&\tau \Bigg(   \tau\bigg(\mu, \tau \Big( \sigma(\mu),\alpha_{x_0} \Big)     \bigg)   ,\tau\Big(\mu,\sigma(\mu)\Big) \Bigg)\\
     &\underset{tr}{=}&\tau \Bigg(   \tau\bigg(\mu, \tau \Big( \sigma(\mu),\alpha_{x_0} \Big)     \bigg)   ,\rho_{x_0} \Bigg)\\
     &\underset{trr}{=}& \tau\bigg(\mu, \tau \Big( \sigma(\mu),\alpha_{x_0} \Big)     \bigg) \\
      &\underset{\sigma(tt)}{=}& \tau\bigg( \tau\Big(\mu,\sigma(\mu)\Big),\alpha_{x_0} \bigg) \\
      &\underset{tr}{=}& \tau\Big(\rho_{x_0},\alpha_{x_0}\Big)  \\
      &\underset{tlr}{=}& [\alpha_{x_0}]\\.
\end{eqnarray*}
 Proceeding in an analogous way, we will obtain that $\mu_{*}\Big(\big[\sigma(\mu)_{*}([\alpha_{y_0}])\big]  \Big)=[\alpha_{y_0}]$. Therefore, have that $\mu_{*}:\pi_1(X,x_0) \longrightarrow \pi_{1}(Y,y_0)$ is a isomorphism, which confirms our proof.  
\end{proof}
Now, consider the following definition:

\begin{definition}[Deformation retract]
 A subspace $A$ of $X$ is called a deformation retract of $X$ if the identity application of $X$ is homotopic to the application which carries all points of $X$ into $A$ in such a way that each point of $A$ remains fixed during homotopy. This is equivalent to the existence of an $H$ application as follows:

 Consider the application $H: X \times I \rightarrow X$ defined by:
 
\end{definition}

$$H(x,t) = \begin{cases}

H(x,0)=x &  \\
H(x,1)= r(x) \\
H(x,t)=x, \forall a \in A.  \\
\end{cases}$$

 Where $r: X \longrightarrow A$ is given by: $r(x)=H(x,1)$, $\forall x\in X$. In the next figure, let $A$ be the curve in red and $X$ all region bounded by the black curve. If $x \in X$ so the application $r(x)$ carrier this point to point in $A$ as show figure 7.

\begin{figure}[!htb]
\centering
\includegraphics[width=0.5\columnwidth]{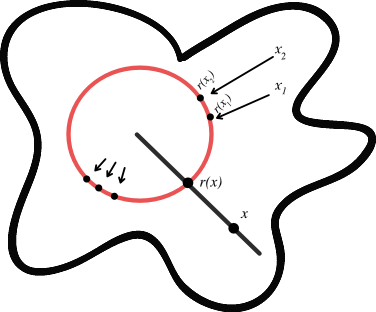}
\caption{Deformation retract of $X$ in $A$ by map $r(x)$. } 
\label{Rotulo10}
\end{figure}

This way we can define $H$ by:

\begin{eqnarray*}
&H&: X \times I \longrightarrow X\\
&H&(x,t) = (1-t)x+tr(x)
\end{eqnarray*}

\begin{lemma}
Let $h,k:(X,x_0) \rightarrow (Y,y_0)$ be continuous and homotopic maps. If $y_0$  is the image of the base point $x_0$ which remains fixed during the homotopy, then $h_{*}=k_{*}$.
\end{lemma}

\begin{theorem}
Let $A$ be a deformation retract of $X$ and $x_0 \in A$. Then the inclusion map $i:(A,x_0)\rightarrow (X,x_0)$ induces an isomorphism of fundamental group.
\end{theorem}

\begin{proof}
Let $r:X\rightarrow A$ be a deformation retract, if we consider the composition $i*r: A \rightarrow A$ so  $(i*r)=Id_{A}$ and therefore  $(i*r)_{*}=i_{*}*r_{*}$ is the identity homomorphism of $\pi_{1}(A,x_0)$. 

On the other hand the composition $r*i:X \rightarrow X$  is not the $Id_{X}$ because $A\subset X$, so consider the map $H: X \times I \longrightarrow X$ defined by:
 
$$H(x,t) = \begin{cases}

H(x,0)=x, \quad \quad  \qquad \forall x \in X. &  \\
H(x,1)= (r*i)(x) \in A, \quad \forall x \in X. \\
H(x_0,t)=x_0, \quad \quad\quad \forall x_0 \in A.  \\
\end{cases}$$ 

So $H$ is a homotopy between $Id_{X}$ and $r*i$, then by lemma 1 we have that  $ (r*i)_{*}= r_{*}*i_{*}$ is the identity homomorphism  of $\pi_1(X,x_0)$.

Therefore, since $r_{*}*i_{*}:\pi_{1}(X,x_0) \rightarrow \pi_{1}(X,x_0)$ is the identity homomorphism, it follows that $i_*$ and $r_*$ are isomorphisms. In other words, the fundamental group of the deformation retract $A$ is isomorphic to the fundamental group of $X$.
\end{proof}

In this work, we wish to determine the fundamental group of a topological space $X$ that is written as the union of two open subsets $U$ and $V$ with both path-connected. If $U\cap V$ is path-connected and let $x_0\in U\cap V$, then the natural morphism $k$ is an isomorphism, that is, the fundamental group of $X$ is the free product of the fundamental groups of $U$ and $V$ with amalgamation of $ \displaystyle\pi_{1}(U\cap V ,x_{0})$.
In what follows, we will need the following theorem to proceed with our main objective:


\begin{theorem}[Seirfet-Van Kampen Theorem]
Let $X=U\cup V$, where $U$ and $V$ are open subsets in $X$, with both path-connected. Suppose that $U \cap V$ are path-connected and let $x_0 \in U\cap V$ be a base point.  The inclusion maps of $U$ and $V$ into $X$ induce group homomorphisms ${\displaystyle j_{1}:\pi _{1}(U,x_{0})\to \pi _{1}(X,x_{0})}$ and ${\displaystyle j_{2}:\pi _{1}(V,x_{0})\to \pi _{1}(X,x_{0})}$. Then $X$ is path connected and ${\displaystyle j_{1}}$ and ${\displaystyle j_{2}}$ form a commutative pushout diagram:

$$
\begin{tikzcd}[column sep=tiny]
& \pi_1(U,x_0) \ar[dr] \ar[drr, "j_1", bend left=20]
&
&[1.5em] \\
\pi_1(U\cap V) \ar[ur, "i_1"] \ar[dr, "i_2"']
&
& \pi_1(U) \ast_{ \pi_1(U\cap V)} \pi_1(V) \ar[r, dashed, "k"]
& \pi_1(X) \\
& \pi_1(V) \ar[ur]\ar[urr, "j_2"', bend right=20]
&
&
\end{tikzcd}
$$

Then, the natural morphism $k$ is an isomorphism, that is, the fundamental group of $X$ is the free product of the fundamental groups of $U$ and $V$ with amalgamation of ${\displaystyle \pi _{1}(U\cap V,x_{0})}$.
\end{theorem}

\subsection{Fundamental Group of the Klein bottle - $\pi_{1}(\mathbb{K}^2,x_{0})$}

 Our objective here is compute the fundamental group of the Klein bottle. We will prove it using computational paths, the results obtained in the previous subsection and the Van Kampen Theorem as we will show in the sequel. 
 
 Consider $\mathbb{K}^{2}$ as the surface known as Klein bottle and the point $x_{0}\in \mathbb{K}^{2}$ as in the following figure.

\begin{figure}[H]
\centering
\includegraphics[width=0.3\columnwidth]{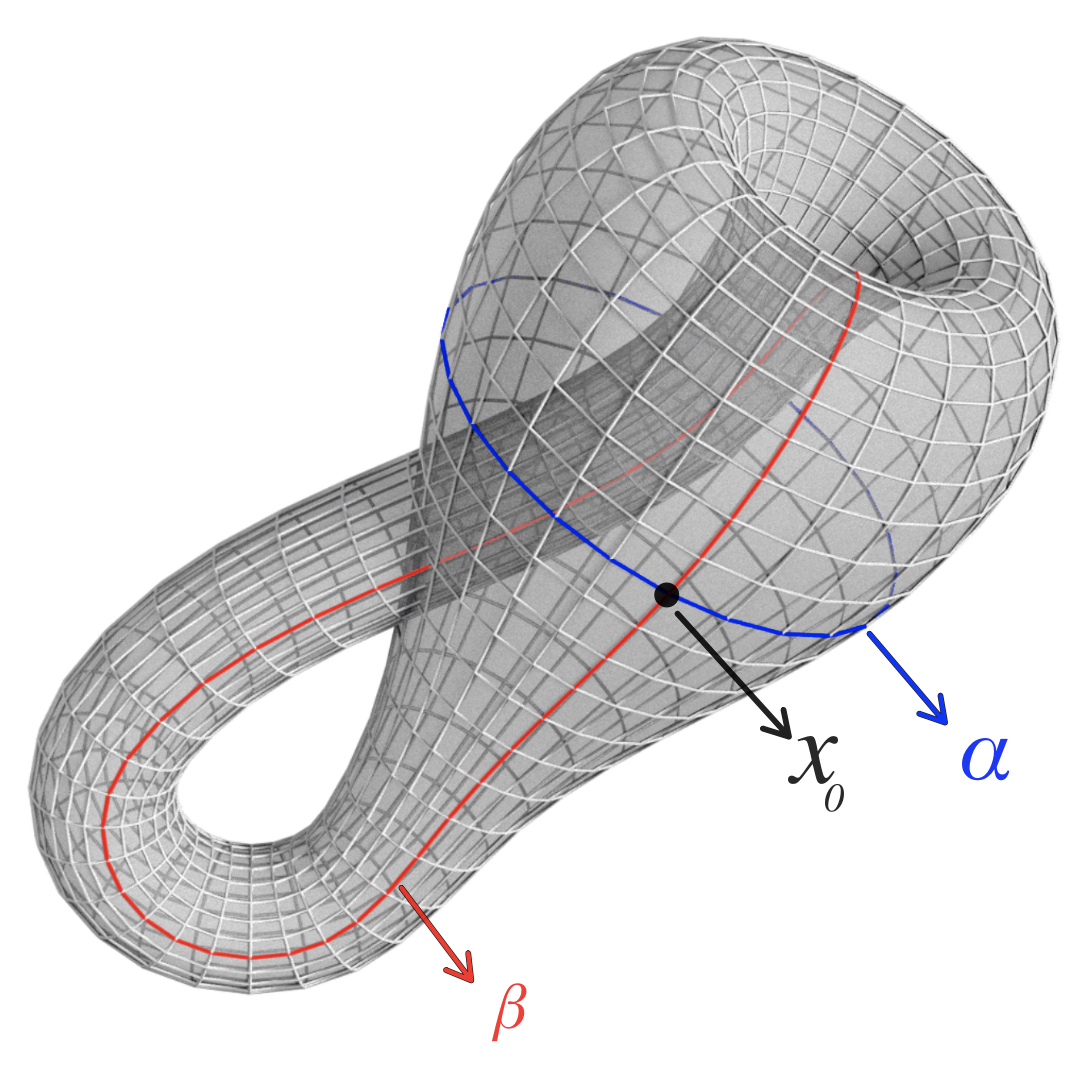}
\caption{ Paths (loops) $\alpha$ and $\beta$ with base point $x_0$ in $\mathbb{K}^{2}$  } 
\label{Rotulo11}
\end{figure}

Given a point $x_{0}$, we can slice the Klein bottle
and represent it as a square whose sides are the loops $\alpha$ and $\beta$, how show in figure 5.

\begin{figure}[H]
\centering
\includegraphics[width=0.3\columnwidth]{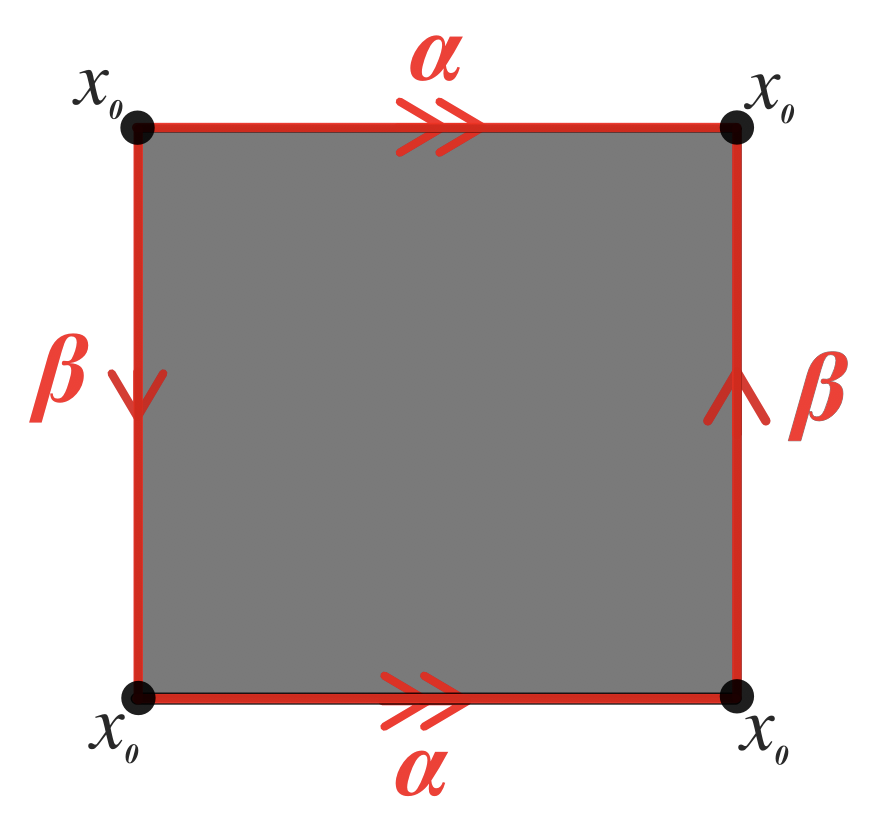}
\caption{Sliced Klein bottle with oriented paths $\alpha$ and $\beta$ } 
\label{Rotulo12}
\end{figure}

The figure 5 is a usual representation of Klein bottle where $\alpha$ and $\beta$ are loops with base point in $x_0$, therefore  $\alpha,\beta \in \pi_{1}(\mathbb{K}^{2},x_{0}).$ We need to prove the following theorem

\begin{theorem}
$\pi_{1}(\mathbb{K}^{2},x_{0})$ is a  free group generated by loops  $\alpha$ and $\beta$ such that $\beta\alpha\beta\alpha^{-1}=\rho_{x_0}$, that is, $$\pi_{1}(\mathbb{K}^{2},x_{0})= \big<\alpha,\beta | \beta\alpha\beta\alpha^{-1}\big>.$$
\end{theorem}

For the proof of this theorem we will need the Van Kampen Theorem, so put $K=V\cup W$, where $V$ and $W$ satisfy the hypotheses of the theorem. Let $i^{V}: V\cap W \rightarrow V$ and $i^{W}: V\cap W \rightarrow W$ be the inclusion maps. So, give $x_1 \in V\cap W$, the homomorphism induced by the inclusion maps are:

$$i_{*}^{V}: \pi_{1}(V \cap W,x_1)\longrightarrow \pi_{1}(V,x_1)$$ and

$$i_{*}^{W}: \pi_{1}(V \cap W,x_1)\longrightarrow \pi_{1}(W,x_1).$$


Let's consider the subsets $V$, $W$ and $V\cap W$ as follows:

\begin{figure}[H]
    \centering
    \subfloat[$V\subset K$]{{\includegraphics[width=3.6cm]{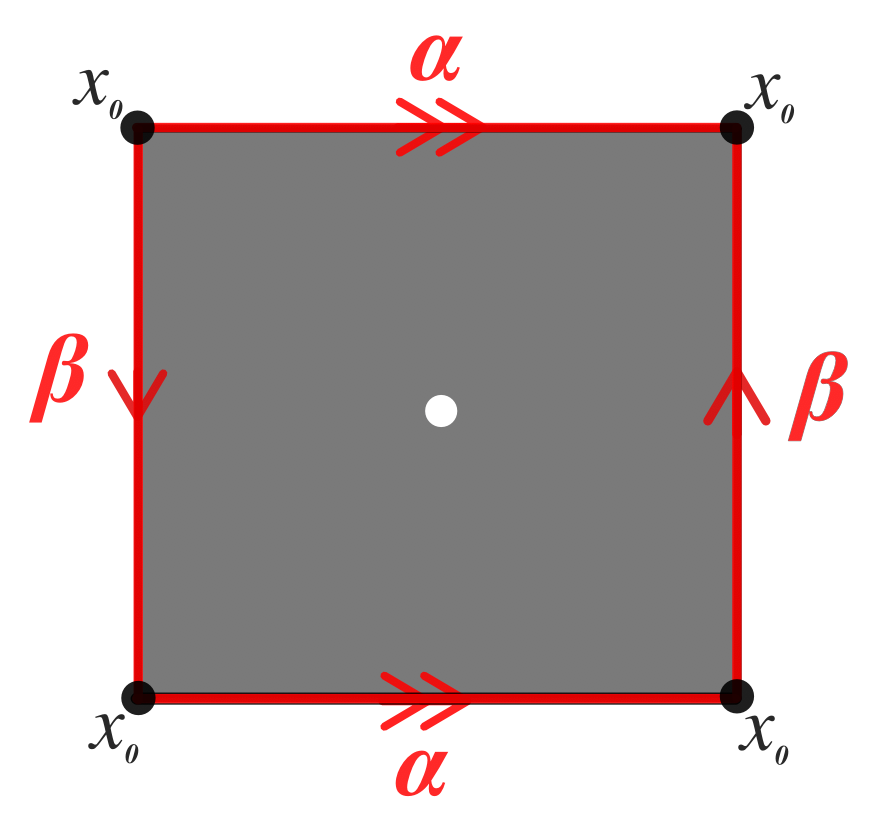} }}%
    \qquad
    \subfloat[$W\subset K$]{{\includegraphics[width=2.8cm]{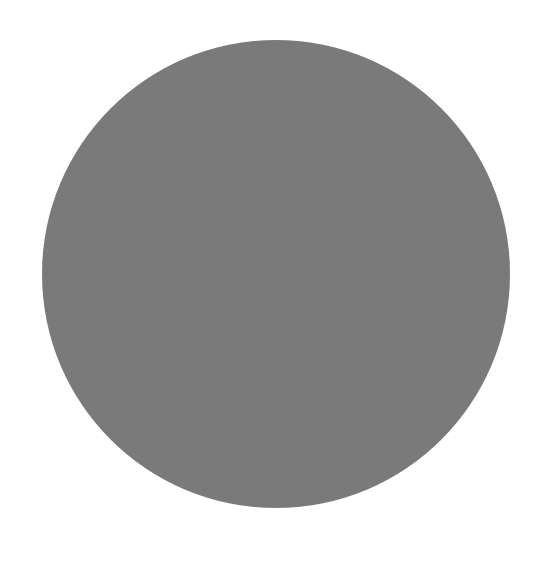} }}%
    \qquad
    \subfloat[$V\cap W\subset K$]{{\includegraphics[width=2.8cm]{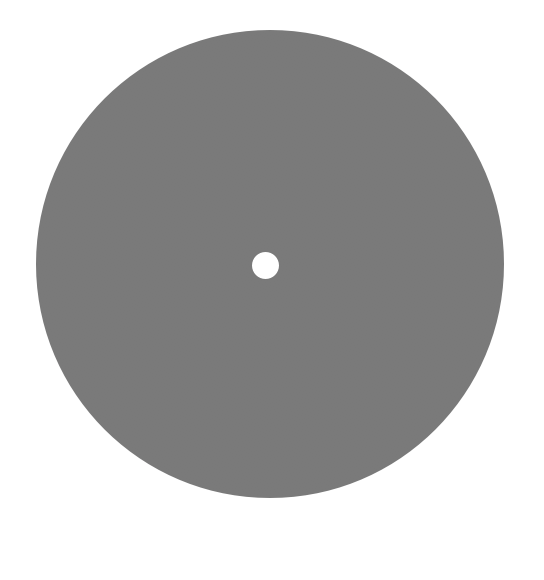} }}%
    \caption{The subsets are represented in darker color.}%
    \label{Rotulo13}
\end{figure}

Notice that $V$ is equal to $K$ minus one point, $W$ is an open disk that covers the puncture and $V\cap W$ is an open disk ( punctured).

Note that only the subset $V$ contains loops $\alpha$ and $\beta$  with base point $x_0$. Now we are interested in calculating the fundamental group of these subsets. We need to work with the same base point for the calculation of the fundamental group of each of the subsets, so we can ensure the homomorphisms induced in the fundamental groups. So consider a point $x_1\in V \cap W$.

Since  $x_1\in W$, we have that all \textit{loop$_{x_1}$} in $W$  is homotopic to constant path  because it continuously deforms to the constant path $\rho$. This way, we can conclude that $\pi_{1}(W,x_1)=\rho_{x_1}$. Now, we calculate the fundamental group of the other subsets.

\begin{itemize}
    
    \item \textsc{\textbf{(i) $\pi_{1}(V \cap W,x_1)$.}}

We are interested in getting the fundamental group of $V \cap W$ at the point $x_1$.
To do this, we need to choose \textit{loop$_{x_1}$} that cannot deform continuously to the point $x_1$ and therefore these loops are those that contain the space hole in its interior. So let $\xi$ be a loop in $V \cap W$ with base point in $x_1$, for simplicity denoted by $\xi_{x_1}$, as in the figure 5. 
 
 \begin{figure}[H]
\centering
\includegraphics[width=0.5\columnwidth]{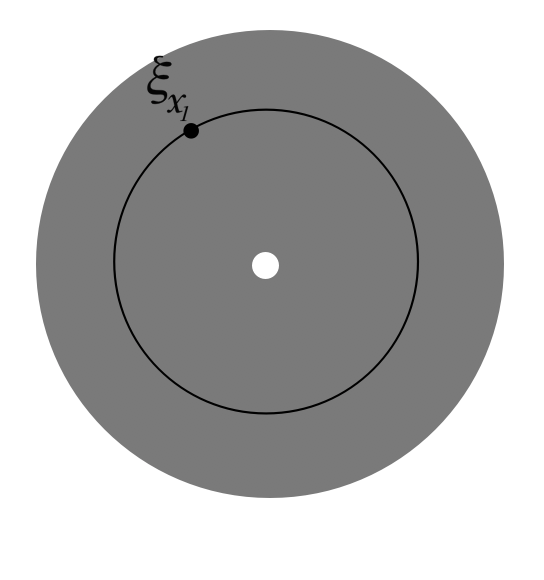}
\caption{ $V \cap W$ with loop $\xi_{x_1}$.} 
\label{Rotulo14}
\end{figure}

 This way we have that  $\xi_{x_1} \in \pi_{1}(V \cap W,x_1)$, we can define by $L$ the subspace of $V \cap W$ of the all loops generated by  $\xi_{x_1}$, that is, $L=\big< \xi_{x_1}\big>$.
 
 Define the map $r:V\cap W \rightarrow L$ such that $r(x)\in L$. $r$ is a map that carries a point $x\in V\cap W$ to the unique intersection point between $L$ 
and the line determined by the hole in the center space with the point $x$, as it is shown in the next figure.
 
 \begin{figure}[H]
\centering
\includegraphics[width=0.5\columnwidth]{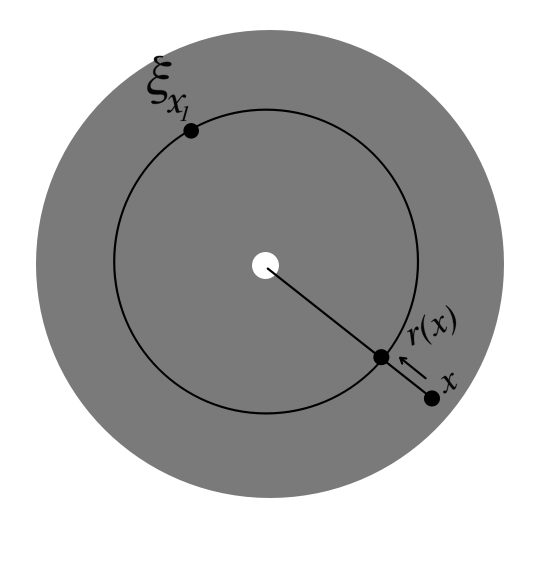}
\caption{ $r$ carry $x$ to $L$. } 
\label{Rotulo15}
\end{figure}

Therefore, we have that $L$ is a deformation retract of $V\cap W$ because there is a map $H: V\cap W \times I \longrightarrow V\cap W$ defined by:
 $$H(x,t) = (1-t)x+tr(x)$$
 
 that satisfies,
 
$$H(x,t) = \begin{cases}

H(x,0)=x &  \\
H(x,1)= r(x) \\
H(x,t)=x, \forall a \in A.  \\
\end{cases}$$ 

that is, $H$ is a homotopy between $r(x)$ and the identity application in such a way that each point of $L$ remains fixed during homotopy. Since $L$ is a deformation retract of $V\cap W$ we have by inclusion map $i^{V\cap W}:L \rightarrow V\cap W$ induces an isomorphism between $\pi_{1}(L,x_1)$ and $\pi_{1}(V\cap W,x_1)$. Since we have that $ \pi_{1}(L,x_1)$ is isomorphic to $\pi_{1}(S^1,x_1)$, we have:
$$\pi_{1}(V\cap W,x_1) \simeq \pi_{1}(L,x_1) \simeq \pi_{1}(S^1,x_1).$$ 

Finally, we can conclude that $\pi_{1}(V\cap W,x_1)= \big<\xi_{x_1}\big>$, that is, the free group generated by $\big<\xi_{x_1}\big>$.

 \item \textsc{\textbf{(ii) $\pi(V,x_0)$.}}

Let $\Tilde{B}$ be the border of $K$, that is, $\Tilde{B}\subset I^2$ is the subspace whose elements are the four loops with base point in $x_0$. We can then define the projection map given by:

\begin{eqnarray*}
Proj: &\Tilde{B}& \longrightarrow B\subset V\\
&\alpha_{x_0}& \longrightarrow Proj(\alpha_{x_0})=\alpha_{x_0}\\
&\beta_{x_0}& \longrightarrow Proj(\beta_{x_0})=\beta_{x_0},
\end{eqnarray*}

where $B=proj(\Tilde{B})\subset V$ is the image of the projection map, as we can see in the next figure:

 \begin{figure}[H]
\centering
\includegraphics[width=0.7\columnwidth]{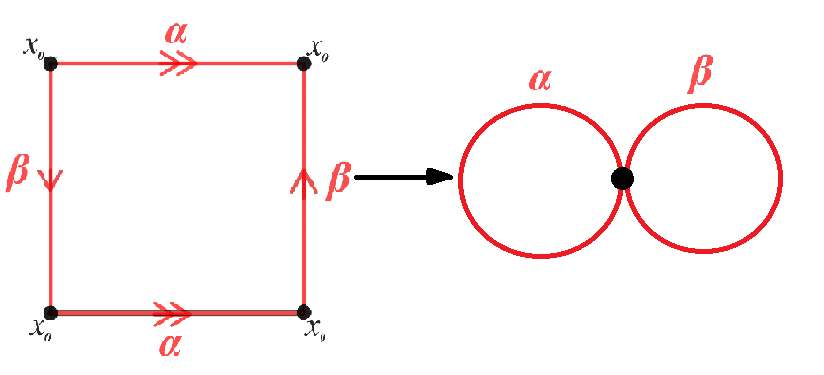}
\caption{ Projection map. } 
\label{Rotulo16}
\end{figure}
Since we can identify the four base points $x_{0}\in\Tilde{B}$ as a single point in the quotient space $B$, we have that $B$ is the subspace of $V$ formed by the collage of two circles (loops$_{x_0}$) by the point $x_0$. So the $\pi_{1}(E,x_0)$  is the free subgroup generated by $[\alpha_{x_0}]$ and $[\beta_{x_0}]$ both elements of $\pi_{1}(E,x_0)$. 
Taking the construction analogous to the case of the intersection made earlier, if we define a map

\begin{eqnarray*}
&H&: V \times I \longrightarrow V\\
&H&(x,t) = (1-t)x+tr(x),
\end{eqnarray*}

we have a homotopy between the map $V$ and the deformation retract $r:V \rightarrow E \subset V$, that is, $H$ induces a homotopy betwwen $V$ on the quaiciente space $E$. So  $\pi_{1}(V,x_0) \simeq \pi_{1}(E,x_0)$ and we can conclude that   $$\pi_{1}(V,x_0)=\big<\alpha_{x_0},\beta_{x_0}\big>,$$ that is, is the free group generated $[\alpha_{x_0}]$ and $[\beta_{x_0}]$.

\end{itemize}

Note that we calculate the $\pi_{1}(V,x_0)$ with respect to the point $x_0$, but to obtain the induced homomorphism $i_{*}^{V}:\pi_{1}(V\cap W,x_1) \longrightarrow \pi_{1}(V,x_1)$ we must obtain $ \pi_{1}(V,x_1)$.

Our problem will be solved by putting $X=V$ in theorem 1, so we can make the following statement:

Since $V$ is a space connected by paths and $x_{0}\underset{\mu}{=} x_{1}$, we can claim that the map  $$\kappa_{\mu}: \pi_1(V,x_1) \longrightarrow \pi_1(V,x_0)$$ given by: $$\kappa_{\mu}([\alpha_{x_1)}])=[\mu]*[\alpha_{x_1}]*[\sigma(\mu])= \tau \big(\tau \big(\mu,\alpha_{x_1}\big),\sigma(\mu) \big),$$

is the desired isomorphism and therefore $\pi_1(V,x_1) \simeq \pi_1(V,x_0)$. In this case, we conclude that: $\kappa_{\mu}([\alpha_{x_1)}])=\beta\alpha\beta\alpha^{-1}$, where $\beta\alpha\beta\alpha^{-1}\in \pi_1(V,x_0)$.

Geometrically we have:

 \begin{figure}[H]
\centering
\includegraphics[width=0.5\columnwidth]{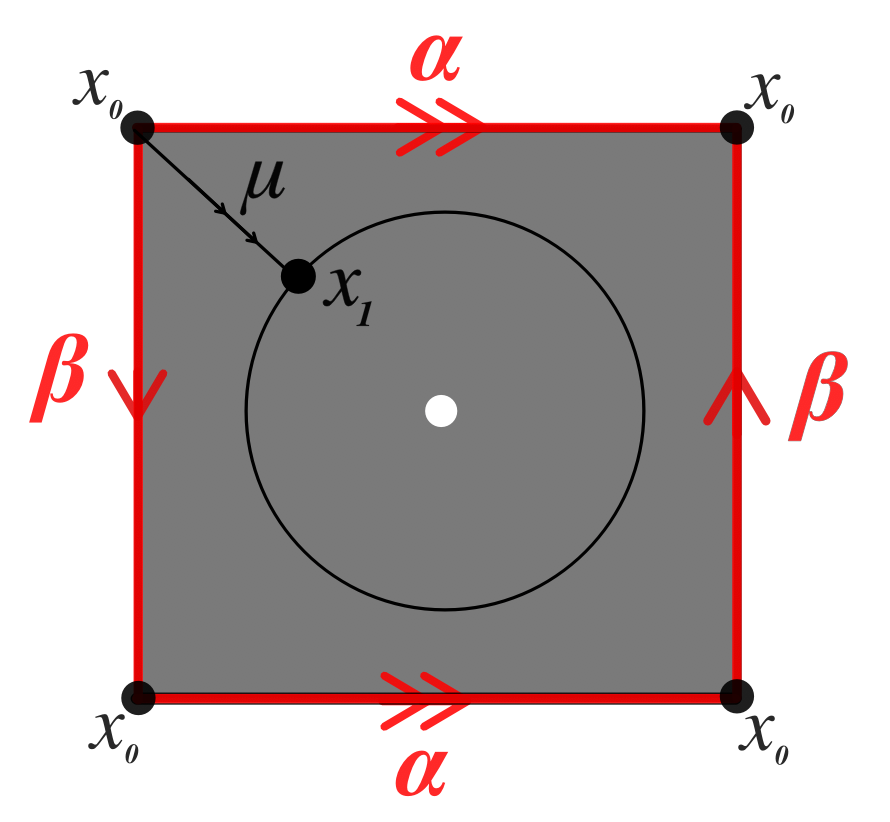}
\caption{Every loop in $x_1$  continuously deforms in  $\beta\alpha\beta\alpha^{-1}$.} 
\label{Rotulo17}
\end{figure}
Therefore  $$\pi_1(V,x_0)=\big<\alpha,\beta \mid \beta\alpha\beta\alpha^{-1}\big>.$$

Once $V \subset K$ so $i^{K}:V \longrightarrow K$ is the inclusion map, we have that the map $i^{K}_{*}:\pi_{1}(V,x_{1}) \longrightarrow \pi_1(K, x_1)$ is the homomorphism induced by inclusion map and your kernel is given by normal subgroup of $\pi_1(V,x_1)$ which is generated by image of $\pi_1(V\cap W,x_1)$. In addition, as $\pi_{1}(V,x_{1})=\rho_{x_1}$ the Seirfet-Van Kampen Theorem ensures that  $\pi_{1}(K,x_{1})$ depends on $\pi_1(V\cap W,x_1)$,  $\pi_1(V,x_1)$ and the morphisms between them. This gives us:

$$
\begin{tikzcd}
\pi_1(V\cap W,x_1) \arrow[r, "i_{*}^{V}" blue] 
    & \pi_1(V,x_1) \arrow[r, "i_{*}^{K}" blue]
    & \pi_1(K,x_1).\\
\end{tikzcd}
$$

On the other hand, as $V$ and $K$ are connected by paths follows from the theorem 3 that:

 $$\kappa_{\mu}: \pi_1(V,x_1) \longrightarrow \pi_1(V,x_0)$$
 
 and
 
  $$\kappa_{\mu}:\pi_1(K,x_1) \longrightarrow \pi_1(K,x_0)$$

are isomorphisms. This way we get the following diagram:

$$
\begin{tikzcd}
\pi_1(V\cap W,x_1) \arrow[r, "i_{*}^{V}" blue] 
    & \pi_1(V,x_1) \arrow[d,  "\kappa_{\mu}" red] \arrow[r, "i_{*}^{K}" blue]
    & \pi_1(K,x_1) \arrow[d, "\kappa_{\mu}" red]\\
    & \pi_1(V,x_0) \arrow[r,  "i_{*}^{K}" blue]
    &\pi_1(K,x_0).
\end{tikzcd}\\\\
$$

Since $i^{K}_{*}:\pi_{1}(V,x_{0}) \longrightarrow \pi_1(K, x_0)$ is a homomorphism induced by inclusion map $i^{K}:V \longrightarrow K$ and by diagram commutativity we have that $i^{K}_{*}:\pi_{1}(V,x_{0}) \longrightarrow \pi_1(K, x_0)$ is surjective with kernel is the normal subgroup generated by the image of $i^{V}_{*}*\kappa_{\mu}=\tau(i^{K}_{*},\kappa_{\mu})$.

Therefore, we can conclude that 

$$\pi_{1}(K, x_0) \simeq \big<\alpha,\beta \mid \beta\alpha\beta\alpha^{-1}\big>.$$


\subsection{Fundamental Group of the Torus  $\pi_1(\mathbb{T}^{2},x_0)$}
The calculations made in the last two subsections give us an idea of how to evaluate the calculations of two interesting surfaces. The torus $\mathbb{T}^2$ is the first surface we can think of to calculate, this choice is very natural since the torus has a slice representation almost identical to that of the Klein bottle, which enables us to obtain a fundamental torus group by means of an approach very similar to that made in  the fundamental group of the Klein bottle. 

But wouldn't this be ``more of the same"? The answer is no! Unlike  $\mathbb{K}^2$, $\mathbb{T}^2$ is not an abstract surface. Furthermore, we can consider defining the torus surface as a type and thus use computational paths to get some relevant results from this, as follows.
  
We now give the formal definition of the torus in homotopy type theory:

\begin{definition}
 The torus $\mathbb{T}^{2}$ is a type generated by:

 \item[ (i)] A base point $x_0:\mathbb{T}^2$
 \item[(ii)] Two paths $\alpha$ and $\beta$ such that: $x_{0}\underset{\alpha}{=} x_{0}:\mathbb{T}^2$ \hspace{0.2cm} and \hspace{0.2cm}  $x_{0}\underset{\beta}{=} x_{0}:\mathbb{T}^2$.
 
 \item[(iii)] One path $co$ that establishes $\beta\alpha\underset{co}{=}\alpha\beta$, i.e., a term $co:Id(\beta\alpha,\alpha\beta)$.

\end{definition}

The first thing one should notice is that this definition does not use only the points of the type $\mathbb{T}^2$ , but also a computational path \textit{loop$_{x_0}$}  between those points and paths between paths. That is why it is called a higher inductive type (Univalent Foundations Program, 2013). Our approach differs from the classic one mainly in the fact that we do not need to simulate the path-space between those points, since computational paths exist in the syntax of our theory. 

Thus, if one starts with a path $x_{0}\underset{\alpha}{=} x_{0}:\mathbb{T}^2$ and/or  $x_{0}\underset{\beta}{=} x_{0}:\mathbb{T}^2$ , one can naturally obtain additional paths applying the path-axioms corresponding to $\rho$, $\tau$ and $\sigma$.  In the original formulation of identity types in type theory, the existence of those additional paths comes from establishing that the paths should be freely generated by the constructors (Univalent Foundations Program, 2013). In our approach, we do not have to appeal to this kind of argument, since all paths come naturally from direct applications of the axioms.

 In homotopy theory, the fundamental group is the one formed by all equivalence classes up to homotopy of paths (\textit{loop$_{x_0}$}) starting from a point $x_0$ and also ending at $x_0$. Since we use computational paths as the syntax counterpart in type theory of homotopic paths, we use it to propose the following definition:
 
\begin{definition}
		$\Pi_{1}(A, x_0, x_0 \underset{loop}{=} x_0)$ is a structure defined as follows:
		
		\begin{center}
			$\Pi_{1}(X, x_0) = \Big\{x_0\underset{r}{=} x_0: X\  \mid r\underset{rw}{=}[loop^n]_{rw}$, \text{for some} $n \in \mathbb{Z} \Big\},$ 
		\end{center}
		where $X$ is a type and $x_0 \underset{loop}{=} x_0$ is a base computational path that generates $x_0 \underset{r}{=} x_0$.
	\end{definition}
	
	For simplicity, we denote by $\Pi_{1}(X, x_0)$ every time we refer to structure $\Pi_{1}(X, x_0, x_0 \underset{loop}{=} x_0)$. 




Since the fundamental groups are obtained by studying the loops, we will be interested in working with \textit{loops} that are not homotopic to the base point $x_{0}$, like \textit{loops} $\alpha$ and $\beta$. These loops will be the generators of $\mathbb{T}^{2}$ as shown in \textbf{figure 9}.

\begin{figure}[H]
\centering
\includegraphics[width=0.6\columnwidth]{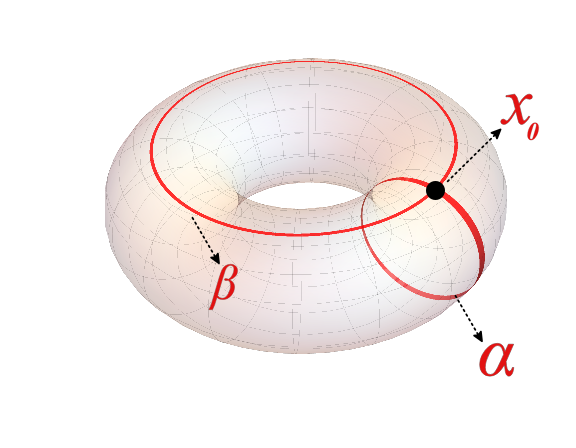}
\caption{Paths $\alpha$ and $\beta$ with base point $x_{0}$ in Torus} 
\label{Rotulo19}
\end{figure}

\begin{definition}[{vertical loop}]
  We define as vertical loop the path (loop) that passes through the inner part of $\mathbb{T}^{2}$ in the vertical direction. In \textbf{figure 9}, this loop is denoted by $\alpha$.
   
\end{definition}

\begin{definition}[{horizontal loop}]
We define as horizontal loop the path (loop) that passes the inner part of $\mathbb{T}^{2}$ in the horizontal direction. In \textbf{figure 9}, this \textit{loop} is denoted by $\beta$.

\end{definition}

Note that these two \textit{loops} are not of the type $\rho$ (homotopic to constant $x_{0}$). Furthermore, we will prove that they generate $\mathbb{T}^{2}$. In what follows, we define and denote by: $\alpha^{n}=$\textit{loop$^{n}_{v}$} the path composed by $n$ vertical loops and by $\beta^{m}=$\textit{loop$^{m}_{h}$} the path composed by $m$ horizontal loops.

Given a point $x_{0}$, we can slice the Torus and represent it as a rectangle whose sides are the loops $\alpha$ and $\beta$, as it is shown in \textbf{figure 10}.

\begin{figure}[H]
\centering
\includegraphics[width=0.4\columnwidth]{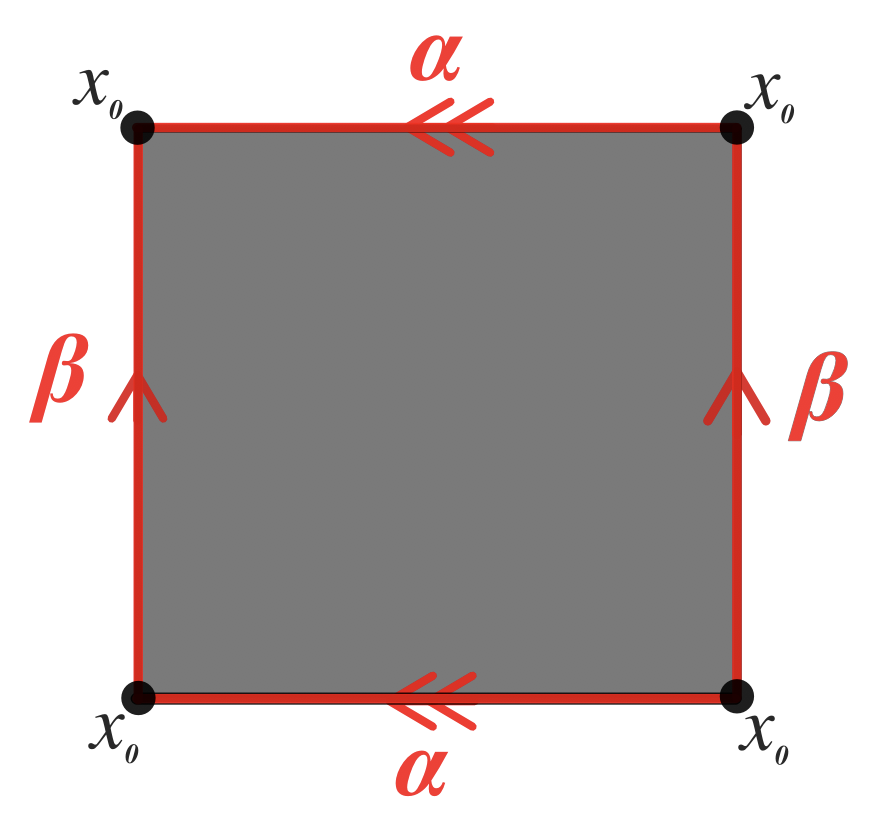}
\caption{Sliced Torus with oriented paths $\alpha$ and $\beta$. } 
\label{Rotulo20}
\end{figure}

Consider the following path in the figure: $\tau \left (\beta,\alpha,\sigma( \beta),\sigma( \alpha )\right)= \beta* \alpha * \beta^{-1} * \alpha^{-1}$:

\begin{proposition}
The aforementioned path is $rw$-equal to the reflexive path.
\end{proposition}

\begin{proof}
Indeed,

\begin{eqnarray*}
\beta*\alpha *\beta^{-1}*\alpha^{-1} &\underset{co}{=}& \beta*\beta^{-1} *\alpha * \alpha^{-1} \\
 &=& \tau \left(\tau(\sigma(\beta),\beta),\tau(\sigma(\alpha)\alpha,)\right)  \\
 &\underset{tst}{=}& \tau(\rho_{x_0},\rho_{x_0})\\
 &\underset{trr}{=}&\rho_{x_0}\\
\end{eqnarray*} 

 and thus,  $\tau \left (\beta,\alpha,\sigma( \beta),\sigma( \alpha )\right)= \beta* \alpha * \beta^{-1} * \alpha^{-1} = \rho_{x_0}$ \\
\end{proof}


 \begin{lemma}
 
All paths in $\mathbb{T}^{2}$ are generated by the application of $\rho$,$\tau$ and $\sigma$ in base path  $[loop^1]_{rw}$. And these paths are $rw$-equal to a path  $[loop^n]_{rw}$, for $n \in \mathbb Z$.

\end{lemma}

\begin{proof}
Consider the following cases:

\item[(i)] Base case: $\beta^{0}\alpha^{0}=\rho_{x_0}$.
\item[(ii)] $\rho_{x_0}\circ \alpha=\tau(\alpha,\rho_{x_0}) \underset{trr}{=}\alpha =\beta^{0}\alpha^{1}.$
\item[(iii)] $\rho_{x_0}\circ \beta=\tau(\beta,\rho_{x_0}) \underset{trr}{=}\beta =\beta^{1}\alpha^{0}.$
\item[(iv)] $\rho_{x_0}\circ \alpha^{-1}=\tau(\sigma(\alpha),\rho_{x_0}) \underset{trr}{=}\sigma(\alpha) =\beta^{0}\alpha^{-1}.$
\item[(v)] $\rho_{x_0}\circ \beta^{-1}=\tau(\sigma(\beta),\rho_{x_0}) \underset{trr}{=}\sigma(\beta) =\beta^{-1}\alpha^{0}.$

Assuming, by the induction hypothesis, that every path is \textit{rw-equal}  $\beta^{n} \alpha^{m}$, we have:

\item[(1)] $\rho_{x_0}\circ \beta^{n}\alpha^{m}=\tau(\beta^{n}\alpha^{m},\rho_{x_0}) \underset{trr}{=}\beta^{n}\alpha^{m}.$
\item[(2)] $\alpha\circ \beta^{n}\alpha^{m}\underset{co}{=} \alpha\circ \alpha^{m}\beta^{n}=\alpha^{m+1}\beta^{n}\underset{co}{=}\beta^{n}\alpha^{m+1}.$
\item[(3)]  $\beta\circ \beta^{n}\alpha^{m}=\beta^{n+1}\alpha^{m} =\beta^{n}\alpha^{m+1}.$

\item[(4)] $\beta^{-1}\circ \beta^{n}\alpha^{m} =(\beta^{-1}\circ (\beta\circ \beta^{n-1}))\alpha^{m}\underset{tt}{=}((\beta^{-1}\circ \beta)\circ \beta^{n-1})\alpha^{m}\underset{tsr}{=}(\rho_{x_0}\circ\beta^{n-1})\alpha^{m}=\beta^{n-1}\alpha^{m}.$

\item[(5)] $\alpha^{-1}\circ \beta^{n}\alpha^{m}\underset{co}{=} \alpha^{-1}\circ \alpha^{m}\beta^{n}=(\alpha^{-1}\circ (\alpha\circ \alpha^{m-1}))\beta^{n}\underset{tt}{=} ((\alpha^{-1}\circ \alpha)\circ \alpha^{m-1})\beta^{n}\underset{tsr}{=}(\rho_{x_0}\circ\alpha^{m-1})\beta^{n}=\alpha^{m-1}\beta^{n} \underset{co}{=}\beta^{n}\alpha^{m-1}.$

 This lemma shows that every path of the fundamental group can be represented by a path of the form \textit{loop$_{x_0}$}$=\beta^{n} \alpha^{m}$, with $m,n \in \mathbb{Z}$.\\
\end{proof}
Now, we will prove that this structure is in fact a group.
\begin{proposition}
$ \left( \Pi_{1}(\mathbb{T}^{2},x_{0}),\circ \right)$ is a group.
\end{proposition}
\smallskip
\begin{proof}

\item [(+): Sum]  $$$$

\begin{prooftree}
\AxiomC{$x_{0} \underset{\beta^{u}\alpha^{v}}{=}x_{0}$}
\AxiomC{$x_{0} \underset{\beta^{r}\alpha^{s}}{=} x_{0}$}
\BinaryInfC{ $x_{0} \underset{\tau \left(\beta^{u}\alpha^{v},\beta^{r}\alpha^{s}\right)}{=}x_{0}$}
\end{prooftree}

But, 

\begin{eqnarray*}
\tau(\beta^{u}\alpha^{v},\beta^{u}\alpha^{v}) &=& (\beta^{r}\alpha^{s}) \circ (\beta^{u}\alpha^{v}) \\
 &=& \beta^{r}\alpha^{s} \beta^{u}\alpha^{v}\\
 &\underset{co}{=}& \beta^{r}\beta^{u}\alpha^{s} \alpha^{v}\\
 &=& \beta^{n}\alpha^{m} \in \Pi_{1}\left(T,x_{0}\right).
\end{eqnarray*}

\item[($\sigma$): Inverse]$$$$ 

\begin{prooftree}
\AxiomC{ $x_{0} \underset{\beta^{n}\alpha^{m}}{=}x_{0}$}
\AxiomC{$x_{0} \underset{\sigma{(\beta^{n})}\sigma{(\alpha^{m})}}{=} x_{0}$}
\BinaryInfC{ $x_{0} \underset{\tau \left(\beta^{n}\alpha^{m},\sigma{(\beta^{n})}\sigma{(\alpha^{m})}\right)}{=}x_{0}$}
\end{prooftree}

But, 

\begin{eqnarray*}
\tau(\beta^{n}\alpha^{m},\sigma{(\beta^{n})}\sigma{(\alpha^{m})}) &=& (\sigma{(\beta^{n})}\sigma{(\alpha^{m})}) \circ (\beta^{n}\alpha^{m}) \\
 &=& \sigma{(\beta^{n})}\sigma{(\alpha^{m})} \beta^{n}\alpha^{m}\\
 &\underset{co}{=}&  \sigma{(\beta^{n})}\beta^{n}\sigma{(\alpha^{m})}\alpha^{m}\\
 &\underset{tsr}{=}&\rho_{\beta}\rho_{\alpha}\underset{trr}{=}\rho_{x_{0}}.
\end{eqnarray*}

On the other hand, we have:

\begin{prooftree}
\AxiomC{ $x_{0} \underset{\sigma{(\beta^{n})}\sigma{(\alpha^{m})}}{=}x_{0}$}
\AxiomC{$x_{0} \underset{\beta^{n}\alpha^{m}}{=} x_{0}$}
\BinaryInfC{ $x_{0} \underset{\tau \left(\sigma{(\beta^{n})}\sigma{(\alpha^{m}),\beta^{n}\alpha^{m}}\right)}{=}x_{0}$}
\end{prooftree}

But, 

\begin{eqnarray*}
\tau(\sigma{(\beta^{n})}\sigma{(\alpha^{m})},\beta^{n}\alpha^{m}) &=& (\beta^{n}\alpha^{m}) \circ (\sigma{(\beta^{n})}\sigma{(\alpha^{m})})\\
 &=& \beta^{n}\alpha^{m}\sigma{(\beta^{n})}\sigma{(\alpha^{m})}\\
 &\underset{co}{=}& \beta^{n}\sigma{(\beta^{n})}\alpha^{m}\sigma{(\alpha^{m})}\\
 &\underset{tr}{=}&\rho_{\beta}\rho_{\alpha}\underset{trr}{=}\rho_{x_{0}}.
\end{eqnarray*}


\item[($\epsilon$): Identity]$$$$ 

\begin{prooftree}
\AxiomC{ $x_{0} \underset{\beta^{n}\alpha^{m}}{=}x_{0}$}
\AxiomC{$x_{0} \underset{\rho_{x_{0}}}{=} x_{0}$}
\BinaryInfC{ $x_{0} \underset{\tau \left(\beta^{n}\alpha^{m},\rho_{x_{0}}\right)}{=}x_{0}$}
\end{prooftree}

But, 

\begin{eqnarray*}
\tau(\beta^{n}\alpha^{m},\rho_{x_{0}}) &=& (\rho_{x_{0}}) \circ (\beta^{n}\alpha^{m})  \\
 &=& \rho_{x_{0}}\beta^{n}\alpha^{m}  \\
 &\underset{tlr}{=}&  \beta^{n}\alpha^{m}
\end{eqnarray*}
 
 and so \[
 \tau(\beta^{n}\alpha^{m},\rho_{x_{0}}) \underset{trr}{=} \beta^{n}\alpha^{m}.
\]

On the other hand, we have: $$$$

\begin{prooftree}
\AxiomC{ $x_{0} \underset{\rho_{x_{0}}}{=}x_{0}$}
\AxiomC{$x_{0} \underset{\beta^{n}\alpha^{m}}{=} x_{0}$}
\BinaryInfC{ $x_{0} \underset{\tau \left(\rho_{x_{0},\beta^{n}\alpha^{m}}\right)}{=}x_{0}$}.
\end{prooftree}

\bigskip
But, 

\begin{eqnarray*}
\tau \left(\rho_{x_{0}},\beta^{n}\alpha^{m}\right) &=&  (\beta^{n}\alpha^{m}) \circ (\rho_{x_{0}})  \\
 &=& \beta^{n}\alpha^{m} \rho_{x_{0}}  \\
 &\underset{trr}{=}&  \beta^{n}\alpha^{m}
\end{eqnarray*}

and so \[
 \tau(\rho_{x_{0}},\beta^{n}\alpha^{m}) \underset{trr}{=} \beta^{n}\alpha^{m}.
\]

\item[( $\circ$ ): Associativity]$$$$

\begin{prooftree}
\AxiomC{$x_{0} \underset{\beta^{n}\alpha^{m}}{=}x_{0}$}
\AxiomC{$x_{0} \underset{\beta^{i}\alpha^{j}}{=} x_{0}$}
\BinaryInfC{$ x_{0} \underset{\tau \left(\beta^{n}\alpha^{m},\beta^{i}\alpha^{j}\right)}{=}x_{0}$}
\AxiomC{$x_{0} \underset{\beta^{r}\alpha^{s}}{=}x_{0}$}
\BinaryInfC{$x_{0} \underset{\tau \left(\tau \left(\beta^{n}\alpha^{m},\beta^{i}\alpha^{j}\right),\beta^{r}\alpha^{s} \right)}{=}x_{0}$}.
\end{prooftree}

But, 

\begin{eqnarray*}
\tau \left(\tau \left(\beta^{n}\alpha^{m},\beta^{i}\alpha^{j}\right),\beta^{r}\alpha^{s} \right)  
 &=& (\beta^{r}\alpha^{s}) \circ \tau (\beta^{n}\alpha^{m},\beta^{i}\alpha^{j}) \\
 &=& (\beta^{r}\alpha^{s}) \circ (\beta^{i}\alpha^{j}\circ \beta^{n}\alpha^{m})\\
 &=& (\beta^{r}\alpha^{s}) \circ (\beta^{i}\alpha^{j}\beta^{n}\alpha^{m})\\
 &=& \beta^{r}\alpha^{s} \beta^{i}\alpha^{j}\beta^{n}\alpha^{m}.\\
\end{eqnarray*}
\bigskip

On the other hand, we have:

\begin{prooftree}
\AxiomC{$x_{0} \underset{\beta^{n}\alpha^{m}}{=}x_{0}$}
\AxiomC{$x_{0} \underset{\beta^{i}\alpha^{j}}{=} x_{0}$}
\AxiomC{$x_{0} \underset{\beta^{r}\alpha^{s}}{=} x_{0}$}
\BinaryInfC{$ x_{0} \underset{\tau \left(\beta^{i}\alpha^{j},\beta^{r}\alpha^{s}\right)}{=}x_{0}$}
\BinaryInfC{$x_{0} \underset{\tau (\beta^{n}\alpha^{m}, \tau \left(\beta^{i}\alpha^{j},\beta^{r}\alpha^{s})\right)}{=}x_{0}$}.
\end{prooftree}

But, 

\begin{eqnarray*}
\tau \left(\beta^{n}\alpha^{m} ,\tau \left(\beta^{i}\alpha^{j},\beta^{r}\alpha^{s}\right)\right)  
 &=&  \tau (\beta^{i}\alpha^{j},\beta^{r}\alpha^{s})\circ(\beta^{n}\alpha^{m}) \\
 &=& (\beta^{r}\alpha^{s} \circ \beta^{i}\alpha^{j})\circ (\beta^{n}\alpha^{m}) \\
 &=& (\beta^{r}\alpha^{s}\beta^{i}\alpha^{j})\circ (\beta^{n}\alpha^{m}) \\
 &=& \beta^{r}\alpha^{s} \beta^{i}\alpha^{j}\beta^{n}\alpha^{m}.\\
\end{eqnarray*}

Therefore, it follows that $ \left( \Pi_{1}(\mathbb{T}^{2},x_{0}),\circ \right)$ is a group.

\end{proof}

Now, using an analogous approach as in the previous section we will prove using computational paths and Van Kampen Theorem the following  theorem:

\begin{theorem}
$\pi_{1}(\mathbb{T}^{2},x_{0})$ is a  free group generated by loops  $\alpha$ and $\beta$ such that $\beta\alpha\beta^{-1}\alpha^{-1}=\rho_{x_0}$, that is, $$\pi_{1}(\mathbb{T}^{2},x_{0})\simeq \big<\alpha,\beta | \beta\alpha\beta^{-1}\alpha^{-1}\big>.$$
\end{theorem}

For the proof of this theorem we will need the Van Kampen Theorem, so put $\mathbb{T}^2=V\cup W$, where $V$ and $W$ satisfy the hypotheses of the theorem. Let $i^{V}: V\cap W \longrightarrow V$ and $i^{W}: V\cap W \longrightarrow W$ be the inclusion maps. So, give $x_1 \in V\cap W$, the homomorphism induced by the inclusion maps are:

$$i_{*}^{V}: \pi_{1}(V \cap W,x_1)\longrightarrow \pi_{1}(V,x_1)$$ and

$$i_{*}^{W}: \pi_{1}(V \cap W,x_1)\longrightarrow \pi_{1}(W,x_1).$$

Let's consider the subsets $V$, $W$ and $V\cap W$ as follows:

\begin{figure}[H]
    \centering
    \subfloat[$V\subset \mathbb{T}^2$]{{\includegraphics[width=3.6cm]{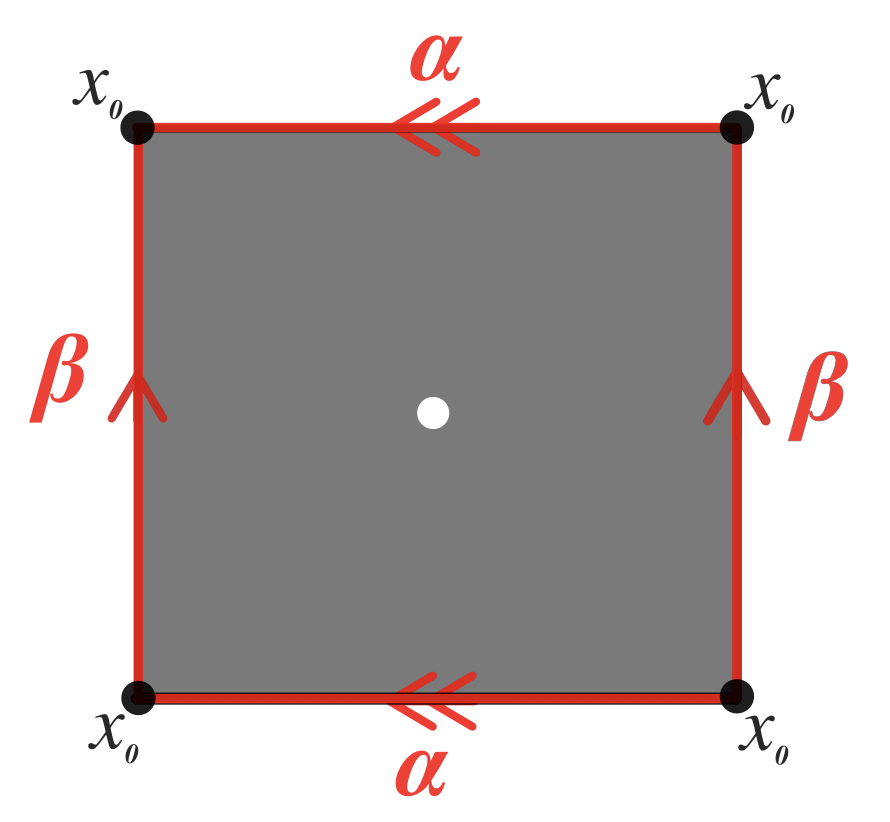} }}%
    \qquad
    \subfloat[$W\subset \mathbb{T}^2$]{{\includegraphics[width=2.8cm]{Klein4.png} }}%
    \qquad
    \subfloat[$V\cap W\subset \mathbb{T}^2$]{{\includegraphics[width=2.8cm]{Klein5.png} }}%
    \caption{The subsets are represented in darker color.}%
    \label{Rotulo21}
\end{figure}

From previous results we have:

\begin{itemize}
    \item [(i)] $\pi_{1}(W,x_1)=\rho_{x_1}$, where $x_1\in W$.
    \item [(ii)]  $\pi_{1}(V\cap W,x_1)= \big<\xi_{x_1}\big>$, where $\xi_{x_1}$ is a loop in $x_1$ that contains the hole in its interior.
    
    \item[(iii)]  $\pi_{1}(V,x_0)= \big<\alpha_{x_0},\beta_{x_0}\big>$, that is, is the free group generated $[\alpha_{x_0}]$ and $[\beta_{x_0}]$. Here we also use the arguments of the point identification $x_0$ and the projection application.
\end{itemize}

\begin{itemize}
    \item [(iv)] Notice, again, that we calculate the $\pi_{1}(V,x_0)$ with respect to the point $x_0$, but to obtain the induced homomorphism $i_{*}^{V}:\pi_{1}(V\cap W,x_1) \longrightarrow \pi_{1}(V,x_1)$ we must obtain $ \pi_{1}(V,x_1)$. By  theorem 1, since $V$ is a space connected by paths and $x_{0}\underset{\mu}{=} x_{1}$, we can claim that the map  $$\kappa_{\mu}: \pi_1(V,x_1) \longrightarrow \pi_1(V,x_0)$$ given by: $$\kappa_{\mu}([\alpha_{x_1)}])=[\mu]*[\alpha_{x_1}]*[\sigma(\mu])= \tau \big(\tau \big(\mu,\alpha_{x_1}\big),\sigma(\mu) \big),$$

is the desired isomorphism and therefore $\pi_1(V,x_1) \simeq \pi_1(V,x_0)$. In this case, we conclude that: $\kappa_{\mu}([\alpha_{x_1)}])=\beta\alpha\beta\alpha^{-1}$, where $\beta\alpha\beta\alpha^{-1}\in \pi_1(V,x_0)$.
\end{itemize}

Geometrically, we have:

 \begin{figure}[H]
\centering
\includegraphics[width=0.5\columnwidth]{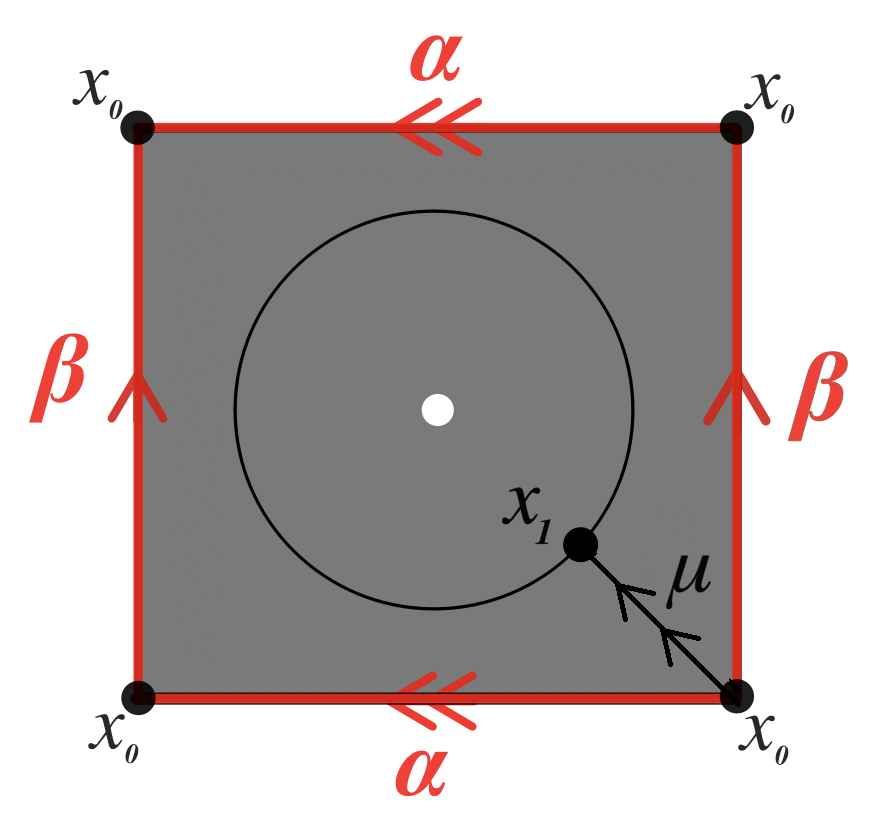}
\caption{Every loop in $x_1$  continuously deforms in $ \beta\alpha\beta^{-1}\alpha^{-1}$.} 
\label{Rotulo22}
\end{figure}
Therefore  $$\pi_1(V,x_0)= \big<\alpha,\beta \mid \beta\alpha\beta^{-1}\alpha^{-1}\big>.$$

Since $V \subset \mathbb{T}^2$ so $i^{\mathbb{T}^2}:V \longrightarrow \mathbb{T}^2$ is the inclusion map, we have that the map $i^{\mathbb{T}^2}_{*}:\pi_{1}(V,x_{1}) \longrightarrow \pi_1(\mathbb{T}^2, x_1)$ is the homomorphism induced by inclusion map and its kernel is given by normal subgroup of $\pi_1(V,x_1)$ which is generated by image of $\pi_1(V\cap W,x_1)$. In addition, as $\pi_{1}(V,x_{1})=\rho_{x_1}$ the Seirfet-Van Kampen Theorem ensures that  $\pi_{1}(\mathbb{T}^2,x_{1})$ depends on $\pi_1(V\cap W,x_1)$,  $\pi_1(V,x_1)$ and the morphims between them. This gives us:

$$
\begin{tikzcd}
\pi_1(V\cap W,x_1) \arrow[r, "i_{*}^{V}" blue] 
    & \pi_1(V,x_1) \arrow[r, "i_{*}^{\mathbb{T}^2}" blue]
    & \pi_1(\mathbb{T}^2,x_1).\\
\end{tikzcd}\\\\
$$

On the other hand, since $V$ and $\mathbb{T}^2$ are connected by paths, it follows from theorem 3 that:

 $$\kappa_{\mu}: \pi_1(V,x_1) \longrightarrow \pi_1(V,x_0)$$
 
 and
 
  $$\kappa_{\mu}:\pi_1(\mathbb{T}^2,x_1) \longrightarrow \pi_1(\mathbb{T}^2,x_0)$$

are isomorphisms. This way we get the following diagram:

$$
\begin{tikzcd}
\pi_1(V\cap W,x_1) \arrow[r, "i_{*}^{V}" blue] 
    & \pi_1(V,x_1) \arrow[d,  "\kappa_{\mu}" red] \arrow[r, "i_{*}^{\mathbb{T}^2}" blue]
    & \pi_1(\mathbb{T}^2,x_1) \arrow[d, "\kappa_{\mu}" red]\\
    & \pi_1(V,x_0) \arrow[r,  "i_{*}^{\mathbb{T}^2}" blue]
    &\pi_1(\mathbb{T}^2,x_0).
\end{tikzcd}\\\\
$$

Since $i^{\mathbb{T}^2}_{*}:\pi_{1}(V,x_{0}) \longrightarrow \pi_1(\mathbb{T}^2, x_0)$ is a homomorphism induced by inclusion map $i^{\mathbb{T}^2}:V \longrightarrow \mathbb{T}^2$ and by diagram commutativity we have that $i^{\mathbb{T}^2}_{*}:\pi_{1}(V,x_{0}) \longrightarrow \pi_1(\mathbb{T}^2, x_0)$ is surjective with kernel is the normal subgroup generated by the image of $i^{V}_{*}*\kappa_{\mu}=\tau(i^{\mathbb{T}^2}_{*},\kappa_{\mu})$.

Therefore, we can conclude that 

$$\pi_{1}(\mathbb{T}^2, x_0) \simeq \big<\alpha,\beta \mid \beta\alpha\beta^{-1}\alpha^{-1}\big>.$$


  \subsection{Fundamental Group of Two-holed Torus - $\pi_1(\mathbb{M}_{2},x_0)$}

 \begin{definition}
 The connected sum of two $n$-dimensional connected surfaces $M_1$ and $M_2$ is the surface $M$, defined up to homeomorphism, obtained by removing an open set homeomorphic to the disk to each of the surfaces and by identifying (``glueing") the boundaries. So we can denote the connected sum by: $$M=M_1\# M_2.$$
   \end{definition}

  We want compute the fundamental group of the connected sum of two torus, the surface resulting from the connected sum ("glueing") of two torus is called two holed torus, denoted by $\mathbb{M}_2=\mathbb{T}_1^2\#\mathbb{T}_2^2$ 
  and can be represented geometrically by:

 \begin{figure}[H]
\centering
\includegraphics[width=0.9\columnwidth]{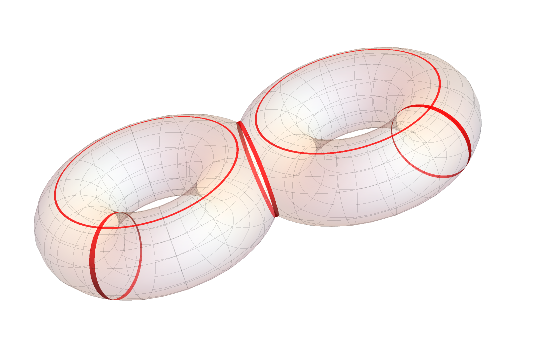}
\caption{$\mathbb{M}_2=\mathbb{T}_1^2\#\mathbb{T}_2^2$ Two holed torus.} 
\label{Rotulo23}
\end{figure}
\begin{proof}

The results obtained in the last subsections mean that our objective can be obtained rather straightforwardly, as we will see next. Consider $\mathbb{M}^2=M_1 \cup M_2$ and see that $M_0=M_1\cap M_2$ is homotopy equivalent to circle $S^1$ and therefore  $\pi_1(\mathbb{M}_0)\simeq \pi_1(S^1)\simeq \pi_1(\mathbb{Z})$.

\begin{figure}[H]
    \centering
    \subfloat{{\includegraphics[width=5.5cm]{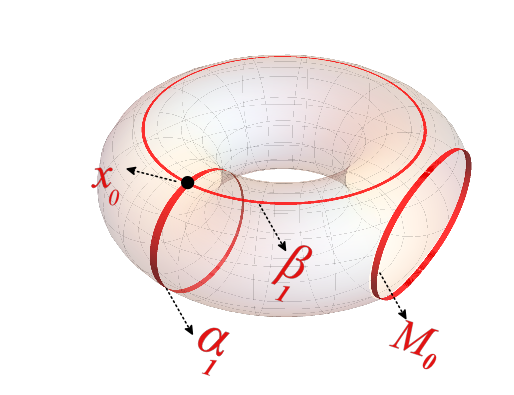} }}%
    \qquad
    \subfloat{{\includegraphics[width=5.5cm]{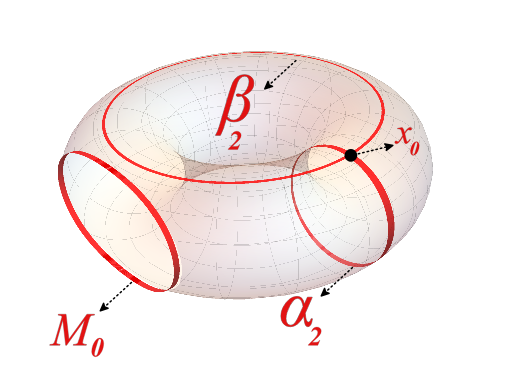} }}%
    \qquad
    \caption{$M_1$ (left) and $M_2$ (right) .} 
    \label{Rotulo25}
\end{figure}

Slicing $M_1$  and $M_2$ , we can represent them as rectangles, whose sides are the loops $\alpha_1$, $\beta_1$ for $M_1$ and $\alpha_2$, $\beta_2$ for $M_2$ as shown in \textbf{figure 14}.

\begin{figure}[H]
    \centering
    \subfloat{{\includegraphics[width=4.0cm]{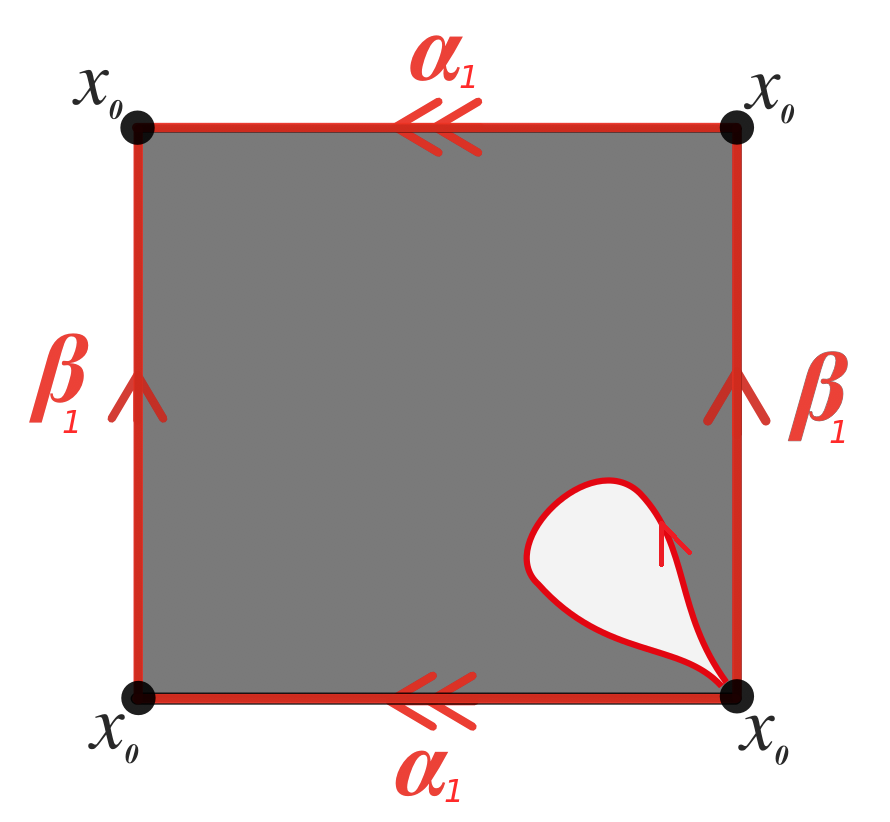} }}%
    \qquad
    \subfloat{{\includegraphics[width=4.0cm]{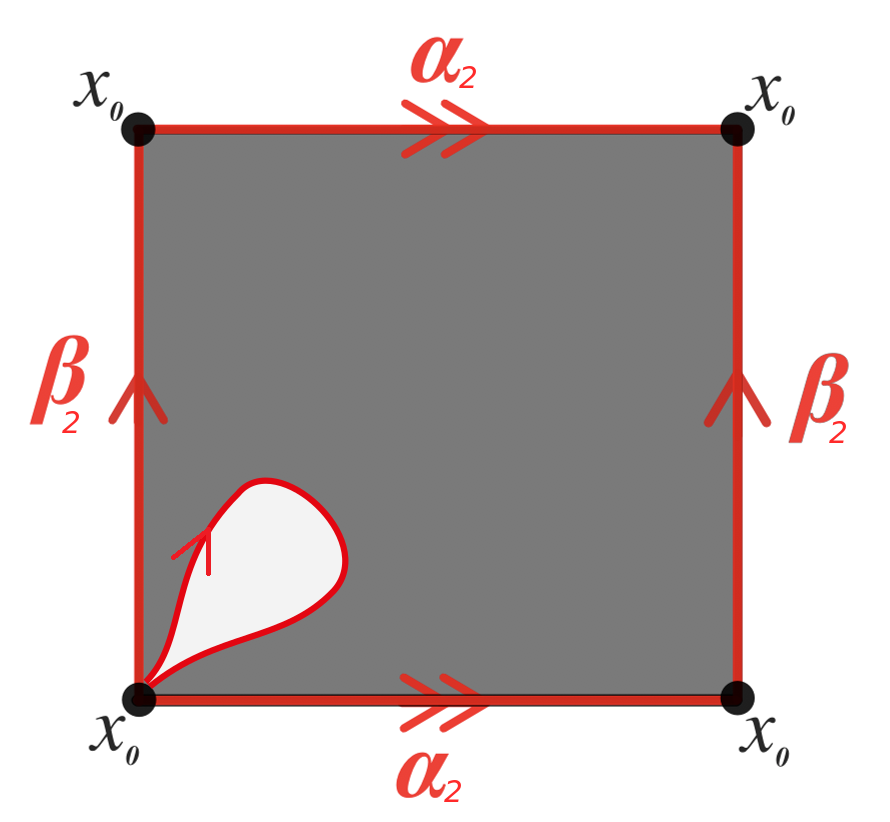} }}%
    \qquad
    \caption{$M_1$ sliced and $M_2$  sliced.} 
    \label{Rotulo25}
\end{figure}

 But $M_i$ is homotopy equivalent to $\mathbb{T}_{i}^{2}$ minus one point $x_{1}$ for $i=\{1,2\}$. So we can calculate $\pi_1\big( \mathbb{M}_i \big)$ by means of the $\pi_1\big(\mathbb{T}_i^2\smallsetminus\{x_1\}\big)$ because
  $$\pi_1(\mathbb{M}_i)\simeq \pi_1(\mathbb{T}_i^2\smallsetminus\{x_1\}).$$

Therefore, we can calculate the fundamental groups of:

\begin{figure}[H]
    \centering
    \subfloat{{\includegraphics[width=4.0cm]{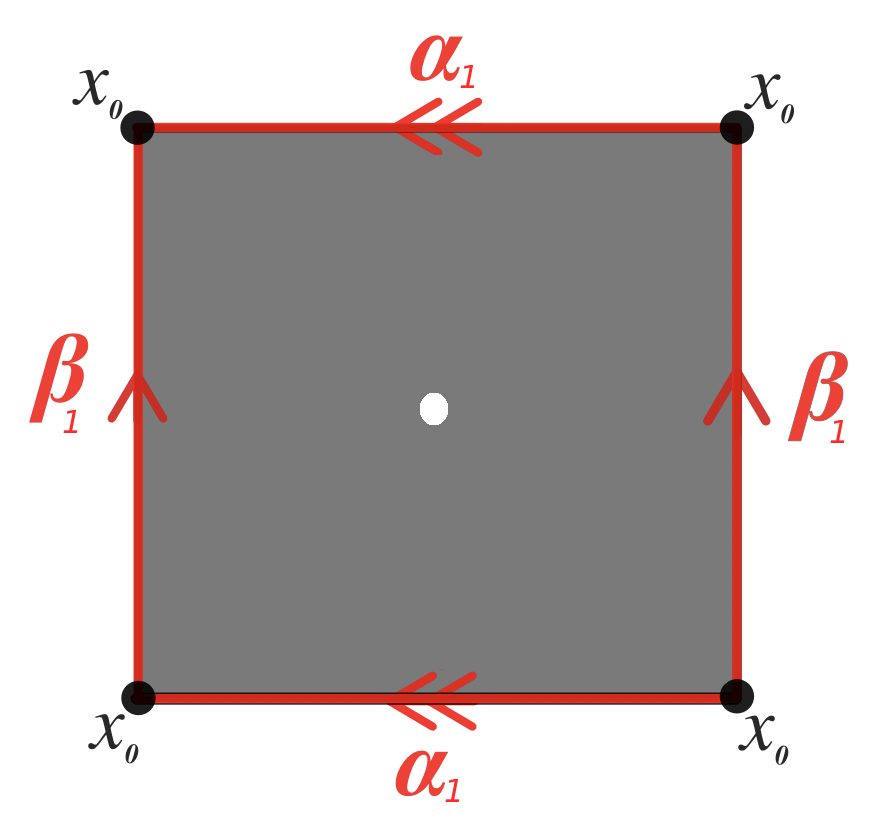} }}%
    \qquad
    \subfloat{{\includegraphics[width=4.0cm]{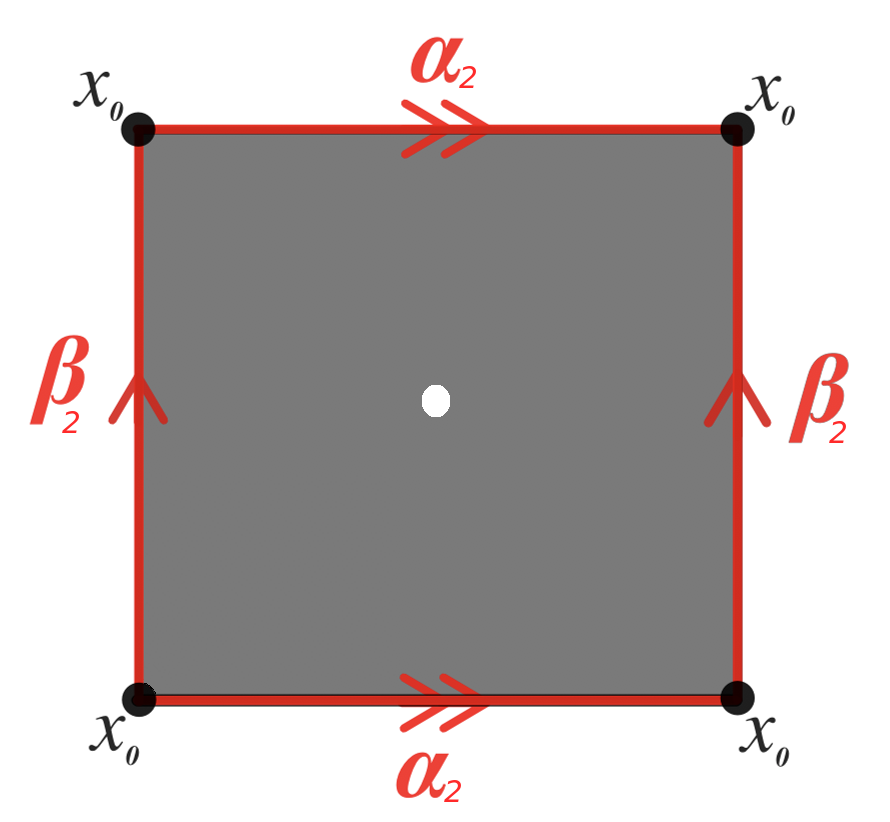} }}%
    \qquad
    \caption{$\mathbb{T}_1^2\smallsetminus\{x_1\}$ (left) and $\mathbb{T}_2^2\smallsetminus\{x_1\}$ (right) sliced.} 
    \label{Rotulo24}
\end{figure}

Again, from previous results we have:

\begin{itemize}
\item [(i)] $\pi_{1}(M_0,x_0)\simeq \mathbb{Z}.$
\item [(ii)]  $\pi_{1}(M_1,x_0)\simeq \big<\alpha_1,\beta_1 \mid [\alpha_1,\beta_1]\big>$, where  $[\alpha_1,\beta_1]=\beta_1\alpha_1\beta_1^{-1}\alpha_1^{-1}=\rho_{x_0}$.
\item[(iii)]  $\pi_{1}(M_2,x_0)\simeq \big<\alpha_2,\beta_2 \mid [\alpha_2,\beta_2]\big>$, where  $[\alpha_2,\beta_2]=\beta_2\alpha_2\beta_2^{-1}\alpha_2^{-1}=\rho_{x_0}$.
\end{itemize}

By Van Kampen theorem, we can conclude:

  $$\pi_{1}(\mathbb{M}_2,x_0)=\pi_{1}(\mathbb{T}^2 \#\mathbb{T}^2,x_0)\simeq \pi_{1}(M_1,x_0)*_{\pi_{1}(M_1 \cap M_2,x_0)} \pi_{1}(M_1,x_0)$$

that is,
$$\pi_{1}(\mathbb{M}_2,x_0) \simeq \big<\alpha_1,\beta_1, \alpha_2,\beta_2, \mid [\alpha_1,\beta_1][\alpha_2,\beta_2]=\rho_{x_0}\big>,$$
where $[\alpha_i,\beta_i]=\beta_i\alpha_i\beta_i^{-1}\alpha_i^{-1}=\rho_{x_0}$ for $i=1,2$.
\end{proof}

This last result allows us to calculate the fundamental group of the $n$-holed torus. If $\mathbb{M}_n=\mathbb{T}^2\#\mathbb{T}^2\#\mathbb{T}^2\#\ldots \#\mathbb{T}^2\#\mathbb{T}^2$, that is, $\mathbb{M}_n$ is the connected sum of $n$ torus so:

$$\pi_{1}(\mathbb{M}_n,x_0)\simeq \big<\alpha_1,\beta_1, \alpha_2,\beta_2,\ldots,\alpha_n,\beta_n \mid [\alpha_,\beta_1][\alpha_2,\beta_2]\ldots[\alpha_n,\beta_n]=\rho_{x_0}\big>,$$

where $[\alpha_i,\beta_i]=\beta_i\alpha_i\beta_i^{-1}\alpha_i^{-1}=\rho_{x_0}$ for $i=1,2, \ldots n$.


	\section{Conclusion}

Our purpose is to further explore the study of results that can be obtained with  a labelled deduction system based on the concept of computational paths (sequence of rewrites). Moreover, we realise that the theory can be used with a computational approach to algebraic topology. In the following, to verify whether the theory of computational paths proves other results (Theorems, propositions) and whether the theories fit to the other definitions is a proposal for a sequel to this work.

\newpage
	\appendix
	
	\section{Subterm Substitution}
	
	In Equational Logic, the sub-term substitution is given by the following inference rule \cite{Ruy2}:
	\begin{center}
		\begin{bprooftree}
			\AxiomC{$s = t$ }
			\UnaryInfC{$s\theta = t\theta$}
		\end{bprooftree}
	\end{center}
	
	One problem is that such rule does not respect the sub-formula property. To deal with that, \cite{chenadec} proposes two inference rules:
	
	\begin{center}
		\begin{bprooftree}
			\AxiomC{$M = N$}
			\AxiomC{$C[N] = O$}
			\RightLabel{$IL$ \quad}
			\BinaryInfC{$C[M] = O$}
		\end{bprooftree}
		\begin{bprooftree}
			\AxiomC{$M = C[N]$}
			\AxiomC{$N = O$}
			\RightLabel{$IR$ \quad}
			\BinaryInfC{$M = C[O]$}
		\end{bprooftree}
	\end{center}
	
	\noindent where M, N and O are terms.
	
	As proposed in \cite{Ruy1}, we can define similar rules using computational paths, as follows:
	
	\begin{center}
		\begin{bprooftree}
			\AxiomC{$x =_r {\cal C}[y]: A$}
			\AxiomC{$y =_s u : A'$}
			\BinaryInfC{$x =_{{\tt sub}_{\tt L}(r,s)} {\cal C}[u]: A$}
		\end{bprooftree}
		\begin{bprooftree}
			\AxiomC{$x =_r w : A'$}
			\AxiomC{${\cal C}[w]=_s u : A$}
			\BinaryInfC{${\cal C}[x]=_{{\tt sub}_{\tt R}(r,s)} u : A$}
		\end{bprooftree}
	\end{center}
	
	\noindent where $C$ is the context in which the sub-term detached by '[ ]' appears and $A'$ could be a sub-domain of $A$, equal to $A$ or disjoint to $A$.
	
	In the rule above, ${\cal C}[u]$ should be understood as the result of replacing every occurrence of $y$ by $u$ in $C$.
	
	\newpage
	
	\section{List of Rewrite Rules}
	
	We present all rewrite rules of $LND_{EQ}$-$TRS$. They are as follows (All have been taken from \cite{Ruy1}):
	
	\begin{itemize}
	 \item[1.] $\sigma(\rho) \triangleright_{sr} \rho$ 
	 \item[2.] $\sigma(\sigma(r)) \triangleright_{ss} r$\
	 \item[3.] $\tau({\cal C}[r] , {\cal C}[\sigma(r)]) \triangleright_{tr}  {\cal C }[\rho]$
	\item [4.] $\tau({\cal C}[\sigma(r)], {\cal C}[r]) \triangleright_{tsr} {\cal C}[\rho]$
	\item[5.] $\tau({\cal C}[r], {\cal C}[\rho]) \triangleright_{trr} {\cal C}[r]$
	\item [6.]$\tau({\cal C}[\rho], {\cal C}[r]) \triangleright_{tlr} {\cal C}[r]$ 
	\item[7.] ${\tt sub_L}({\cal C}[r], {\cal C}[\rho]) \triangleright_{slr} {\cal C}[r]$
	\item [8.] ${\tt sub_R}({\cal C}[\rho], {\cal C}[r]) \triangleright_{srr} {\cal C}[r]$ 
	\item[9.]  ${\tt sub_L} ({\tt sub_L} (s, {\cal C}[r]), {\cal C}[\sigma(r)]) \triangleright_{sls} s$
	\item[10.]  ${\tt sub_L} ( {\tt sub_L} (s , {\cal C}[\sigma(r)]) , {\cal C}[r]) \triangleright_{slss} s$
	\item[11.]  ${\tt sub_R} ({\cal C}[s], {\tt sub_R} ({\cal C}[\sigma(s)],r)) \triangleright_{srs} r$
	\item[12.]  ${\tt sub_R} ({\cal C}[\sigma(s)], {\tt sub_R} ({\cal C}[s] ,  r )) \triangleright_{srrr} r$
	\item[13.] $\mu_1 ( \xi_1 ( r))\triangleright_{mx2l1} r$
	\item[14.]  $\mu_1 ( \xi_\land ( r,s))\triangleright_{mx2l2} r$
	\item[15.] $\mu_2 ( \xi_\land ( r,s))\triangleright_{mx2r1} s$
	\item[16.] $\mu_2 ( \xi_2 ( s))\triangleright_{mx2r2} s$
	\item[17.] $\mu ( \xi_1 (r) , s , u) \triangleright_{mx3l} s$
	\item[18.] $\mu (\xi_2 (r) , s , u) \triangleright_{mx3r} u$
	\item[19.] $\nu (\xi (r)) \triangleright_{mxl} r$
	\item[20.] $\mu (\xi_2 (r) , s) \triangleright_{mxr} s$ 
	\item[21.] $\xi ( \mu_1 (r),\mu_2(r) ) \triangleright_{mx} r$ 
	\item[22.] $\mu ( t, \xi_1 (r), \xi_2 (s)) \triangleright_{mxx} t$  
	\item[23.] $\xi ( \nu (r) ) \triangleright_{xmr} r$  
	\item[24.] $\mu (s,\xi_2 (r)) \triangleright_{mx1r} s$
	\item[25.]  $\sigma(\tau(r,s)) \triangleright_{stss} \tau(\sigma(s),  \sigma(r))$ 
	\item[26.]  $\sigma({\tt sub_L}(r,s)) \triangleright_{ssbl} {\tt sub_R}(\sigma(s), \sigma(r))$ 
	\item[27.]  $\sigma ({\tt sub_R} (r,s)) \triangleright_{ssbr} {\tt sub_L} (\sigma(s),  \sigma (r))$ 
	\item[28.]  $\sigma(\xi (r)) \triangleright_{sx} \xi ( \sigma(r))$
	\item[29.]  $\sigma(\xi (s, r)) \triangleright_{sxss} \xi ( \sigma(s),  \sigma(r))$
	\item[30.]  $\sigma(\mu (r)) \triangleright_{sm} \mu ( \sigma(r))$
	\item[31.]  $\sigma(\mu (s, r)) \triangleright_{smss} \mu (\sigma(s),  \sigma(r))$ 
	\item[32.]  $\sigma(\mu (r,u,v)) \triangleright_{smsss} \mu ( \sigma(r),\sigma(u),\sigma(v))$
	\item[33.]  $\tau (r, {\tt sub_L} (\rho , s)) \triangleright_{tsbll} {\tt sub_L}  (r,s)$
	\item[34.]  $\tau (r, {\tt sub_R} (s, \rho)) \triangleright_{tsbrl}  {\tt sub_L} (r,s)$ 
	\item[35.]  $\tau({\tt sub_L}(r,s),t) \triangleright_{tsblr} \tau (r, {\tt sub_R} (s,t))$
	\item[36.]  $\tau ({\tt sub_R} (s,t),u) \triangleright_{tsbrr} {\tt sub_R} (s, \tau  (t,u))$ 
	\item[37.]  $\tau(\tau(t,r),s) \triangleright_{tt} \tau(t,\tau (r,s)) $
	\item[38.]  $\tau ({\cal C}[u], \tau ({\cal C}[\sigma(u)] , v)) \triangleright_{tts} v$
	\item[39.]  $\tau ({\cal C}[\sigma(u)] , \tau ({\cal C}[u] , v)) \triangleright_{tst} u$

	\end{itemize}
\end{document}